\pdfoutput=1

\documentclass[a4paper,UKenglish,cleveref, autoref, thm-restate, numberwithinsect]{lipics-v2021}

\hideLIPIcs  %

\usepackage{esvect}
\usepackage{tikz}
\usetikzlibrary{arrows.meta}
\usetikzlibrary{arrows}

\title{Homogeneous Algebraic Complexity Theory and Algebraic Formulas} %

\author{Pranjal Dutta}{School of Computing, National University of Singapore (NUS), Singapore \and \url{https://sites.google.com/view/pduttashomepage}}{pranjal@nus.edu.sg}{https://orcid.org/0000-0001-9137-9025}{Funded under the project ``Foundation of Lattice-based Cryptography'', by NUS-NCS Joint Laboratory for Cyber Security.}
\author{Fulvio Gesmundo}{Institut de Math{\'e}matiques de Toulouse, Universit{\'e} Paul Sabatier, Toulouse, France \and \url{https://fulges.github.io/}}{fgesmund@math.univ-toulouse.fr}{https://orcid.org/0000-0001-6402-021X}{}

\author{Christian Ikenmeyer}{University of Warwick, Warwick, UK  \and \url{https://www.dcs.warwick.ac.uk/~u2270030/}}{christian.ikenmeyer@warwick.ac.uk}{https://orcid.org/0000-0003-4654-177X}{Supported by EPSRC grant EP/W014882/1.}

\author{Gorav Jindal}{Max Planck Institute for Software Systems, Saarbr{\"u}cken, Germany \and \url{https://goravjindal.github.io/}}{gjindal@mpi-sws.org}{https://orcid.org/0000-0002-9749-5032}{}

\author{Vladimir Lysikov}{Ruhr-Universit\"at Bochum, Bochum, Germany \and \url{https://qi.rub.de/lysikov}}{Vladimir.Lysikov@ruhr-uni-bochum.de}{https://orcid.org/0000-0002-7816-6524}{Part of the work was was done while V.L. was affiliated with the QMATH Centre, University of Copenhagen. V.L. acknowledges financial support from VILLUM FONDEN via the QMATH Centre of Excellence (Grant No. 10059) and the European Union (ERC Grant Agreements 818761 and 101040907).
Views and opinions expressed are however those of the author(s) only and do not necessarily reflect those of the European Union or the European Research Council Executive Agency. Neither the European Union nor the granting authority can be held responsible for them.}

\authorrunning{P. Dutta, F. Gesmundo, C. Ikenmeyer, G. Jindal, and V. Lysikov} %

\Copyright{Pranjal Dutta, Fulvio Gesmundo, Christian Ikenmeyer, Gorav Jindal, and Vladimir Lysikov} %

\ccsdesc[500]{Theory of computation~Algebraic complexity theory}
\keywords{Homogeneous polynomials, Waring rank, Arithmetic formulas, Border complexity, Geometric Complexity theory, Symmetric polynomials} %

\category{} %

\relatedversion{} %

\nolinenumbers %

\EventEditors{Venkatesan Guruswami}
\EventNoEds{1}
\EventLongTitle{15th Innovations in Theoretical Computer Science Conference (ITCS 2024)}
\EventShortTitle{ITCS 2024}
\EventAcronym{ITCS}
\EventYear{2024}
\EventDate{January 30 to February 2, 2024}
\EventLocation{Berkeley, CA, USA}
\EventLogo{}
\SeriesVolume{287}
\ArticleNo{42}

\newcommand{\bbC}{\mathbb{C}}

\newcommand{\frakS}{\mathfrak{S}}

\newcommand{\IN}{\mathbb{N}}%
\newcommand{\IS}{\mathbb{S}}%
\newcommand{\IC}{\mathbb{C}}%
\newcommand{\eps}{\epsilon}%

\newcommand{\Pb}{\mathcal{B}}

\newcommand{\GL}{\operatorname{GL}}

\newcommand{\idthree}{\mathrm{id}_3}
\newcommand{\idtwo}{\mathrm{id}_2}
\newcommand{\vvirg}{,\ldots,}

\newcommand{\WR}{\mathsf{WR}}
\newcommand{\bwr}{\underline{\mathsf{WR}}}
\newcommand{\bcr}{\underline{\mathsf{CR}}}
\newcommand{\IMM}{\textup{\textsf{IMM}}}

\newcommand{\la}{\lambda}

\newcommand{\nvar}{\mathsf{nvar}}

\newcommand{\poly}{\mathsf{poly}}%
\newcommand{\polylog}{\mathsf{polylog}}%

\newcommand{\VBP}{\mathsf{VBP}}
\newcommand{\VQBP}{\mathsf{VQBP}}
\newcommand{\VP}{\mathsf{VP}}

\newcommand{\VTP}{\mathsf{V3P}}
\newcommand{\VQP}{\mathsf{VQP}}

\newcommand{\VQTP}{\mathsf{VQ3P}}
\newcommand{\VTF}{\mathsf{V3F}}

\newcommand{\VQTF}{\mathsf{VQ3F}}

\newcommand{\VNP}{\mathsf{VNP}}
\newcommand{\VF}{\mathsf{VF}}
\newcommand{\VQF}{\mathsf{VQF}}
\newcommand{\per}{\textup{per}}
\renewcommand{\P}{\mathsf{P}}
\newcommand{\NP}{\mathsf{NP}}

\newcommand{\nc}{\textup{\textsf{nc}}}

\newcommand{\CR}{\mathsf{CR}}

\renewcommand{\det}{\textup{det}}
\newcommand{\sgn}{\textup{sgn}}

\AtBeginDocument{%
\let\oldref\ref
\let\ref\cref
}
\renewcommand{\eqref}[1]{(\oldref{#1})}

\begin{document}

\maketitle

\begin{abstract}
  We study algebraic complexity classes and their complete polynomials under \emph{homogeneous linear} projections, not just under the usual affine linear projections that were originally introduced by Valiant in 1979. These reductions are weaker yet more natural from a geometric complexity theory (GCT) standpoint, because the corresponding orbit closure formulations do not require the padding of polynomials.
We give the \emph{first} complete polynomials for $\VF$, the class of sequences of polynomials that admit small algebraic formulas, under homogeneous linear projections: The sum of the entries of the non-commutative elementary symmetric polynomial in 3 by 3 matrices of homogeneous linear forms.

Even simpler variants of the elementary symmetric polynomial are hard for the topological closure of a large subclass of $\VF$: 
the sum of the entries of the non-commutative elementary symmetric polynomial in 2 by 2 matrices of homogeneous linear forms, and homogeneous variants of the continuant polynomial (Bringmann, Ikenmeyer, Zuiddam, JACM '18). This requires a careful study of circuits with arity-3 product gates.
\end{abstract}

\section{Motivation: Geometric Complexity Theory and Padding}
Geometric Complexity Theory (GCT) is an approach towards proving algebraic variants of the $\P \ne \NP$ conjecture using algebraic geometry and representation theory~\cite{MS01,MS08}.
Let $\det_n:=\Sigma_{\sigma \in \frakS_n}\sgn(\pi)\prod_{i=1}^n x_{i,\sigma(i)}$ be the determinant polynomial, and let $\per_m:=\Sigma_{\sigma \in \frakS_m}\Pi_{i=1}^m x_{i,\sigma(i)}$ be the permanent polynomial.
An algebraic version of the $\P \ne \NP$ conjecture, often called Valiant's \emph{determinant vs.~permanent} conjecture, states that the smallest
size of a matrix $A$ whose entries are affine linear polynomials such that $\det(A)=\per_m$, is {\em not polynomially} bounded in $m$. Mulmuley and Sohoni strengthened the conjecture by allowing the permanent to be approximated arbitrarily closely coefficientwise instead of being computed exactly.

The Mulmuley–Sohoni conjecture can be stated in terms of group orbit closures as $\ell^{n-m}\per_m \not \in \overline{\GL_{n^2}\det_{n}}$, if $n=\poly(m)$; here $\GL_{n^2}:=\GL(\IC^{n\times n})$ acts on the space
of homogeneous degree $n$ polynomials in $n^2$ variables by (invertible) linear transformations of the variables\footnote{For a homogeneous polynomial $p$ and $g\in\GL_{n^2}$ define the homogeneous polynomial $gp$ via $(gp)(\vv x) := p(g^t \vv x)$. The orbit is defined as $\GL_{n^2}p := \{gp \mid g \in \GL_{n^2}\}$.}, $\ell$ is some homogeneous linear polynomial (one can assume $\ell:=x_{1,1}$),
and the closure can be taken equivalently in the Zariski or the Euclidean topology, see e.g.~\cite[AI.7.2 Folgerung]{Kra85}. %
The polynomial $\ell^{m-n}\per_n$ is called the `padded permanent', and the phenomenon of multiplying with a power of a linear form is called \emph{padding}.
Note here that the action of $\GL_{n^2}$ replaces variables by \emph{homogeneous} linear polynomials. One could formulate this setup without padding, but then the reductive group $\GL_{n^2}$ would have to be replaced by the general affine group (see e.g.~\cite{MS21}), which is \emph{not} a reductive group.
For reductive groups, every representation decomposes into a direct sum of irreducible representations.
This is important for the representation theoretic attack proposed in \cite{MS01,MS08}, hence the padding is introduced in those papers.
The idea is that $\ell^{n-m}\per_m \in \overline{\GL_{n^2}\det_{n}}$ if and only if
$\overline{\GL_{n^2}\ell^{n-m}\per_m} \subseteq \overline{\GL_{n^2}\det_{n}}$.
Such an inclusion induces a $\GL_{n^2}$-equivariant surjection between the coordinate rings and between their homogeneous degree $\delta$ components, see e.g.~\cite{BLMW11}:
$
\IC[\overline{\GL_{n^2}\det_{n}}]_\delta \;\twoheadrightarrow \;\IC[\overline{\GL_{n^2}\ell^{n-m}\per_m}]_\delta.
$
Now, since the group $\GL_{n^2}$ is reductive, both sides decompose into irreducible representations of $\GL_{n^2}$:
\[
\underbrace{\IC[\overline{\GL_{n^2}\det_{n}}]_\delta}_{=\bigoplus_\la d_\la V_\la} \;\twoheadrightarrow\;\underbrace{\IC[\overline{\GL_{n^2}\ell^{n-m}\per_m}]_\delta}_{=\bigoplus_\la p_\la V_\la}\;,
\]
where $\la$ is a non-increasing list of $n^2$ many nonnegative integers,
and $V_\la$ is the irreducible $\GL_{n^2}$ representation of type $\la$.
Schur's lemma (see e.g.~\cite{FH91}) implies that $\forall \la : d_\la \geq p_\la$.
A $\la$ with $d_\la < p_\la$ is called a \emph{multiplicity obstruction}. If additionally we have that $d_\la=0$, then $\la$ is called an {\em occurrence obstruction}.
Issues with the padding were known from the beginning, and machinery to carry over information from $\IC[\overline{\GL_{m^2}\per_m}]$ to $\IC[\overline{\GL_{n^2}\ell^{n-m}\per_m}]$ was discussed, see e.g.~\cite{BLMW11}.
The impact of the padding on $\la$ was first highlighted by Kadish and Landsberg~\cite{KL14}, where they use the padding to classify a large class of $\la$ as \emph{not useful}. This was later strengthened in \cite{IP17,BIP19}, where it was shown that all relevant $\la$ have strictly positive $d_\la$, so that occurrence obstructions are not sufficient to prove Mulmuley and Sohoni’s conjecture.
This is known as the occurrence obstruction no-go result.

However, the padding can be removed by replacing $\det_n$ by the iterated matrix multiplication polynomial in $2n+n^2(d-2)$ variables:
\[
\IMM_{n,d} :=
\left(\begin{smallmatrix}x_{1,1,1} & x_{1,2,1} & \cdots & x_{1,n,1}\end{smallmatrix}\right)
\left(\begin{smallmatrix}
x_{1,1,2} &    \cdots   &   x_{1,n,2}     \\
  \vdots    & \ddots &   \vdots     \\
  x_{n,1,2}    &   \cdots    & x_{n,n,2}
\end{smallmatrix}\right)
\cdots
\left(\begin{smallmatrix}
x_{1,1,d-1} &    \cdots   &   x_{1,n,d-1}     \\
  \vdots    & \ddots &   \vdots     \\
  x_{n,1,d-1}    &   \cdots    & x_{n,n,d-1}
\end{smallmatrix}\right)
\left(\begin{smallmatrix}
x_{1,1,d} \\
  \vdots\\
  x_{n,1,d}
\end{smallmatrix}\right).
\]
Again, the task is to show that a surjection cannot exist:
\[
\underbrace{\IC[\overline{\GL_{2n+n^2(d-2)}\IMM_{n,d}}]_\delta}_{=\bigoplus_\la i_\la V_\la} \;\twoheadrightarrow\; \underbrace{\IC[\overline{\GL_{2n+n^2(d-2)}\per_d}]_\delta}_{=\bigoplus_\la j_\la V_\la}.
\]
Analogously to $d_\la$ vs $p_\la$, one searches for $\la$ with $i_\la< j_\la$.
In fact, it is known that the $j_\la$ can be determined independently of $n$ via inheritance theorems (see \cite{Ike13}):
$\IC[\overline{\GL_{d^2}\per_d}]_\delta =\bigoplus_\la j_\la V_\la$.
There are no no-go results known for this approach, but no strong equations vanishing on the orbit closure of $\IMM$ have been found so far.

Our main contribution in this paper is the discovery of new natural polynomials that serve as much simpler replacements for $\IMM$, which are still powerful enough to imply variants of Valiant's conjecture, see \S\oldref{subsec:mainresults}.

\section{Algebraic Complexity Theory}
A sequence of natural numbers $m = (m_n)_{n\in\IN}$ is called \emph{polynomially bounded} if there exists a univariate polynomial~$t$ such that $\forall n \in \IN: m_n\leq t(n)$.
Let $\Pb$ denote the set of all polynomially bounded sequences.
Let $\IS := \IC[x_1,x_2,\ldots]$ denote the set of all polynomials, and let $\IS_d$ denote the vector space of all homogeneous degree $d$ polynomials (including the zero polynomial).
We sometimes use the notation $n\mapsto a(n)$ to denote the function $a$, for example $n\mapsto n$ is the identity map.
For a sequence $g\in\IS^\IN$ of polynomials let $\deg(g) := n\mapsto \deg(g_n)$ be the sequence of degrees.
Analogously, for a polynomial $p$ define $\nvar(p)$ to be the number of variables appearing in $p$, and
for a sequence $g\in\IS^\IN$ of polynomials let $\nvar(g) := n\mapsto \nvar(g_n)$.
A sequence $g \in \IS^\IN$ is called a \emph{p-family} if $\deg(g)\in \Pb$ and $\nvar(g)\in \Pb$.
We sometimes call p-families \emph{ungraded} p-families, and we propose a definition of a \emph{graded} p-family in \S\oldref{sec:gradedpfamilies}, which will be useful for obtaining padding-free orbit closure formulations. The classical complexity classes that we discuss in this section are defined in terms of ungraded p-families.

An algebraic formula is a directed tree with a unique sink vertex. The source vertices are labelled by affine linear combinations of variables, and each internal node of the graph is labelled by either $+$ or $\times$. Nodes
compute polynomials in the natural way by induction.
An algebraic circuit is slightly more general: The underlying digraph is required to be acyclic, but not necessarily a tree.
The size of a circuit/formula is the number of its vertices.
$\VF$ is the class of p-families $(f_n)_{n\in\IN}$, with required formula size of $f_n$ being polynomially bounded.
$\VP$ is the class p-families $(f_n)_{n\in\IN}$, with required circuit size of $f_n$ being polynomially bounded.

Every homogeneous degree $d$ polynomial $p$ can be written as a product
\[
p =
\left(\begin{smallmatrix}\ell_{1,1,1} & \ell_{1,2,1} & \cdots & \ell_{1,n,1}\end{smallmatrix}\right)
\left(\begin{smallmatrix}
\ell_{1,1,2} &    \cdots   &   \ell_{1,n,2}     \\
  \vdots    & \ddots &   \vdots     \\
  \ell_{n,1,2}    &   \cdots    & \ell_{n,n,2}
\end{smallmatrix}\right)
\cdots
\left(\begin{smallmatrix}
\ell_{1,1,d-1} &    \cdots   &   \ell_{1,n,d-1}     \\
  \vdots    & \ddots &   \vdots     \\
  \ell_{n,1,d-1}    &   \cdots    & \ell_{n,n,d-1}
\end{smallmatrix}\right)
\left(\begin{smallmatrix}
\ell_{1,1,d} \\
  \vdots\\
  \ell_{n,1,d}
\end{smallmatrix}\right)
\]
of matrices whose entries are homogeneous linear polynomials.
We define $w(p)$ to be the smallest possible such $n$, and call it the \emph{homogeneous branching program width} of $p$.
For an inhomogeneous polynomial, we define $w(p) := \sum_{d\in\IN}w(p_d)$ to be the sum of the widths of its homogeneous components.
$\VBP$ is the class of p-families whose $w$ is polynomially bounded.

The {\em permanental complexity} of a polynomial $f$ is the smallest $n$ such that $f$ can be written as the permanent of an $n\times n$ matrix of affine linear polynomials.
The class~$\VNP$ consists of all p-families $(f_n)_{n\in\IN}$ for which the permanental complexity is polynomially bounded.

It is known that $\VF \subseteq \VBP \subseteq \VP \subseteq \VNP$~\cite{val79, Toda92}.
The conjectures $\VF \neq \VNP$, $\VBP \neq \VNP$, $\VP \neq \VNP$, are known as \emph{Valiant's conjectures}. Especially $\VBP \neq \VNP$ is known as the \emph{determinant vs permanent} problem.
A sequence $(c_n)_{n \in \IN}$ of natural numbers is called \emph{quasipolynomially bounded} if there exists a polynomial $q$ with $\forall n \geq 2: c_n \leq n^{q(\log_2 n)}$.
In the definitions of $\VF$, $\VBP$, $\VP$, if we change the upper bound on the complexity to ``quasipolynomially bounded'' instead of just ``polynomially bounded'',
then each time we obtain the same class, which we call $\VQP$, see \cite{bur00}.
The conjecture $\VNP \not\subseteq \VQP$ is called Valiant's \emph{extended} conjecture.

\subsection{Border Complexity}
\label{subsec:bordercompl}
The complexity notions mentioned above, such as formula size, circuit size, width $w$, permanental complexity, have an associated \emph{border complexity} variant: A polynomial has border complexity $\leq k$ if it is the limit of polynomials of complexity at most $k$. Here, the limit is taken in the Euclidean topology on the coefficient vector space, see e.g.~\cite{IS22}.
Border complexity measures are usually indicated by an underlined symbol: e.g., $\underline{w}$ is the border homogeneous algebraic branching program width.
Clearly $\underline{w}(p)\leq w(p)$ for all polynomials $p$.

The border complexity analogues of the classical classes are denoted by an overline, e.g., $\overline{\VF}$ is the class of p-families with polynomially bounded border formula complexity\footnote{see \cite{IS22} for the definition of the closure of sets of p-families in general.}.
While Valiant's conjecture states that $w(\per)$ grows superpolynomially ($\VBP\neq\VNP$), the Mulmuley-Sohoni conjecture states that $\underline w(\per)$ grows superpolynomially ($\VNP\not\subseteq\overline{\VBP}$).
The extended Valiant's conjecture states that $w(\per)$ grows superquasipolynomially ($\VNP\not\subseteq\VQP$),
and it is natural to
merge these to the extended Mulmuley-Sohoni conjecture: $\underline w(\per)$ grows superquasipolynomially ($\VNP\not\subseteq\overline{\VQP}$).

Border complexity is an old area of study in algebraic geometry. In theoretical computer
science it was introduced in \cite{BCRL79,Bin80} in the context of fast matrix
multiplication.
In algebraic complexity theory, border complexity was first discussed independently in~\cite{MS01, Bur04}.

\section{Graded p-families and Homogeneous Reductions}
\label{sec:gradedpfamilies}
In this section we generalize known concepts from algebraic complexity theory from univariate to bivariate by adding a degree parameter.
This gives the correct setting for homogeneous linear projections, which is the natural setting for {\em padding-free} geometric complexity theory.
We are very formal in this section, because the readers are used to affine projections, and some steps might seem very subtle.

For the connections between the homogeneous and inhomogeneous setting, see \S\oldref{subsec:homoginhomog}.

As usual, for a set $A$, we identify sequences $a\in A^\IN$ with functions $\IN\to A$, and we write $a_n=a(n)$.
We use the same notation for functions $\IN\times\IN\to A$, i.e., $a_{n,d} = a(n,d)$.

A function $m : \IN \times \IN \to \IN$
is called \emph{bivariately polynomially bounded} if there exists a bivariate polynomial~$t$ such that $\forall (n,d) \in \IN\times\IN: m_{n,d}\leq t(n,d)$.
We propose the following definition of a \emph{graded} p-family in order to work with the weak reduction notion of homogeneous linear projections, which enables padding-free orbit closure formulations.

\begin{definition}
\label{def:collection}
A \emph{graded p-family} $f$ is a map $f: \IN\times\IN \to \IS$
such that
\begin{itemize}
\item $(n,d)\mapsto\nvar(f_{n,d})$ is bivariately polynomially bounded, and
    \item every $f_{n,d}$ is either zero or homogeneous of degree~$d$.

\end{itemize}
\end{definition}

For example, $\IMM(n,d)=\IMM_{n,d}$ is a graded p-family.
The natural reduction notion for graded p-families are {\em homogeneous linear projections}, which are defined as follows.
Suppose $U,W$ are finite dimensional complex vector spaces and $p \in \bbC[U]_d$, $q \in \bbC[W]_d$ are homogeneous degree $d$ (where $d >0$) polynomials. 
We say that $p$ is a homogeneous linear projection of $q$, and write $p \leq_{\textup{homlin}} q$, if $p \in \{q \circ A \mid A \colon U \to W \text{ linear}\}$.
For degree $d=0$ we define that for any nonzero $q$ we have $p \leq_{\textup{homlin}} q$.
For graded p-families $f$ and $h$ we write $f\leq_{\textup{p-homlin}} h$ if there exists $m\in\Pb$ such that
for all $n,d$ we have $f_{n,d}\leq_{\textup{homlin}} h_{m_n,d}$.
The border complexity version is analogous: 
$p \trianglelefteq_{\textup{homlin}} q$, if $p \in \overline{\{q \circ A \mid A \colon U \to W \text{ linear}\}}$, and 
$f\trianglelefteq_{\textup{p-homlin}} h$, if $\exists m\in\Pb \ \forall n,d: f_{n,d}\trianglelefteq_{\textup{homlin}} h_{m_n,d}$. 
If $m$ is only quasipolynomially bounded, we obtain the analogous quasipolynomial variants
$f\leq_{\textup{qp-homlin}} h$
and
$f\trianglelefteq_{\textup{qp-homlin}} h$.

Ungraded p-families $g$ are graded p-families in the natural way, by setting $g_{n,d}$ to be the homogeneous degree $d$ component of $g_n$.
In particular, the permanent can be interpreted in this way as a graded p-family.
This allows us to phrase the four conjectures in this language:
\begin{align*}
\VNP \;=\; \VBP
&
~~~~~~~~~~~~~~~~\textup{ if and only if }
&
\per \;\leq_{\textup{p-homlin}}\; \IMM,
\\
\VNP \;\subseteq\; \overline{\VBP}
&
~~~~~~~~~~~~~~~~\textup{ if and only if } 
&
\per \;\trianglelefteq_{\textup{p-homlin}}\; \IMM,
\\
\VNP \;\subseteq\; \VQP
&
~~~~~~~~~~~~~~~~\textup{ if and only if } 
&
\per \;\leq_{\textup{qp-homlin}}\; \IMM,
\\
\VNP \;\subseteq\; \overline{\VQP}
&
~~~~~~~~~~~~~~~~\textup{ if and only if } 
&
\per \;\trianglelefteq_{\textup{qp-homlin}}\; \IMM.
\end{align*}
Since $\per$ is a p-family of homogeneous polynomials, the question $\per \trianglelefteq_{\textup{p-homlin}} \IMM$ is about the existence of an $m\in \Pb$ such that $\forall d: \per_d \;\trianglelefteq_{\textup{homlin}}\; \IMM_{m(d),d}$.
This has a padding-free orbit closure formulation under the general linear group, which is reductive: %
\[
\per_d \;\trianglelefteq_{\textup{homlin}}\; \IMM_{m_d,d}
\quad \textup{ iff } \quad
\overline{\GL_{d^2}\per_d} \;\subseteq\;
\overline{\GL_{2m_d+m_d^2(d-2)}\IMM_{m_d,d}}\;.
\]
This is the main advantage of using homogeneous linear projections as the reduction notion.
Our main contribution is to replace $\IMM$ by simpler graded p-families that capture $\VF$ or the large subset $\VTF$ of $\VF$; see~\ref{def:v3f} in \S\oldref{subsec:aritythree}.
This has two advantages: The orbit closures become simpler, and the separations from $\VNP$ become easier than $\VBP\neq\VNP$, because $\VTF\subseteq\VF\subseteq\VBP$, while the quasipolynomial versions of $\VTF$, $\VF$, $\VBP$ all coincide with $\VQP$.

\subsection{Main Results}
\label{subsec:mainresults}
Let
$
\nc e_d(X_1, \dots, X_n) \;:=\; \sum_{1 \leq I_1 < I_2 <\dots < I_d \leq n} X_{I_1} \dots X_{I_d}\;,
$
denote the elementary symmetric polynomial in noncommuting variables $X_1,\ldots,X_n$.
Let $L : \IC^{3\times 3}\to \IC$ be the sum of all 9 entries.
Let $\nc e_{3,n,d} := L \circ \nc e_d(A_1,A_2,\ldots,A_n)$, where each $A_i$ is a $3\times 3$ matrix of 9 fresh variables.
We denote by $\nc e_3$ the corresponding graded p-family.
\begin{theorem}\label{thm:nceVFpc}

\begin{align*}
\VNP \;=\; \VF
&
~~~~~~~~~~~~~~~~\textup{ if and only if }
&
\per \;\leq_{\textup{p-homlin}}\; \nc e_3,
\\
\VNP \;\subseteq\; \overline{\VF}
&
~~~~~~~~~~~~~~~~\textup{ if and only if } 
&
\per \;\trianglelefteq_{\textup{p-homlin}}\; \nc e_3,
\\
\VNP \subseteq \VQP
&
~~~~~~~~~~~~~~~~\textup{ if and only if } 
&
\per \;\leq_{\textup{qp-homlin}}\; \nc e_3,
\\
\VNP \;\subseteq\; \overline{\VQP}
&
~~~~~~~~~~~~~~~~\textup{ if and only if } 
&
\per \;\trianglelefteq_{\textup{qp-homlin}}\; \nc e_3.
\end{align*}
\end{theorem}
Note that
$
\per_d \;\trianglelefteq_{\textup{homlin}}\; \nc e_{3,n,d}
\ \textup{ iff } \ 
\overline{\GL_{d^2}\per_d} \;\subseteq\;
\overline{\GL_{9n}\nc e_{3,n,d}}\;.
$
In the border setting, we manage to get the same results even for $\nc e_2$,
we simplify the orbit closure on the right hand side even further by introducing a new class $\VTF \subseteq \VF$ (see \S\oldref{subsec:aritythree}), whose quasipolynomial version is still $\VQP$.
The parity-alternating elementary symmetric polynomial $C_{n,d}$ is defined via
$
C_{n,d} \;:=\; \sum_{(i_1,i_2,\ldots,i_d) \in I} x_{i_1}x_{i_2} \cdots x_{i_d},
$
where $I$ is the set of length $d$ increasing sequences of numbers $i_1<i_2<\ldots<i_d$ from $\{1,\ldots,n\}$ in which for all $j$ the parity of $i_j$ differs from the parity of $i_{j+1}$,
and $i_1$ is odd, in other words, $i_j \equiv j \pmod 2$.

\begin{theorem}
\label{thm:intro:homogeneousVF}

\begin{align*}
\VNP \;\subseteq\; \overline{\VTF}
&
~~~~~~~~~~~~~~~~~~~~~\textup{ $\Longrightarrow$ } 
&
\per \;\trianglelefteq_{\textup{p-homlin}}\; C,
\\
\VNP \;\subseteq\; \overline{\VF}
&
~~~~~~~~~~~~~~~~~~~~~\textup{ $\Longleftarrow$ }
&
\per \;\trianglelefteq_{\textup{p-homlin}}\; C,
\\
\VNP \;\subseteq\; \overline{\VQP}
&
~~~~~~~~~~~~~~~~~\textup{ if and only if } 
&
\per \;\trianglelefteq_{\textup{qp-homlin}}\; C.
\end{align*}
\end{theorem}
Note that
``$
\per_d \;\trianglelefteq_{\textup{homlin}}\; C_{n,d}
\ \textup{ iff } \ 
\overline{\GL_{d^2}\per_d} \;\subseteq\;
\overline{\GL_{n}C_{n,d}}
$''

is a formulation with an intriguingly simple orbit closure.
Moreover, it seems reasonable to try to prove $\VNP \not\subseteq\overline{\VTF}$ or $\VNP \not\subseteq\overline{\VF}$ before proving the more difficult
$\VNP \not\subseteq\overline{\VBP}$. 

\section{Related Concepts}
\subsection{Classical homogeneous complexity measures: Waring rank, Chow rank, tensor rank}
\label{subsec:classical}
In classical algebraic geometry, homogeneous linear projections are the standard way to compare homogeneous polynomials and tensors.

We list some of the classical examples in this subsection.

Given a homogeneous degree $d$ polynomial $f$, the \emph{Waring rank} of $f$, denoted $\WR(f)$, is the smallest $r$ such that there exist homogeneous linear polynomials $\ell_1 \vvirg \ell_r$, with $f = \sum_{i=1}^r \ell_i^d$.

The \emph{border Waring rank} of $f$, denoted $\bwr(f)$, is the smallest $r$ such that $f$ can be written as limit of a sequence of polynomials $f_\eps$ with $\WR(f_\eps) \leq r$.
Given the graded p-family $P_{n,d} := x_1^d+\cdots+x_n^d$, we see that
$\WR(p)\leq r$ iff $p \leq_{\textup{homlin}} P_{n,d}$
and
$\bwr(p)\leq r$ iff $p \trianglelefteq_{\textup{homlin}} P_{n,d}$, which is equivalent to
$p \in \overline{\GL_n P_{n,d}}$, provided $p$ is defined in the variables $x_1,\ldots,x_n$.
Waring rank was studied already in the eighteenth century \cite{Cay:TheoryLinTransformations,Sylv:PrinciplesCalculusForms,Cleb:TheorieFlachen} in the context of invariant theory, with the aim to determine normal forms for homogeneous polynomials. We mention the famous Sylvester Pentahedral Theorem, stating that a generic cubic form in four variables can be written uniquely as sum of five cubes. At the beginning of the twentieth century, the early work on secant varieties in classical algebraic geometry \cite{Palatini:SuperficieAlg,Terr:seganti} implicitly commenced the study of border Waring rank.
In the algebraic complexity theory literature, Waring rank is called the homogeneous $\Sigma\Lambda\Sigma$-circuit complexity.

The \emph{Chow rank} of $f$, denoted $\CR(f)$, is the smallest $r$ such that there exist homogeneous linear polynomials $\ell_{1,1} \vvirg \ell_{r,d}$, with $f = \sum_{i=1}^r \ell_{i,1}\cdots\ell_{i,d}$.
The \emph{border Chow rank} of $f$, denoted $\bcr(f)$, is the smallest $r$ such that $f$ can be written as limit of a sequence of polynomials $f_\eps$ with $\CR(f_\eps) \leq r$.
Given the graded p-family $Q_{n,d} := x_{1,1}\cdots x_{1,d}+\cdots+x_{n,1}\cdots x_{n,d}$, we see that
$\CR(p)\leq r$ iff $p \leq_{\textup{homlin}} Q_{n,d}$
and
$\bcr(p)\leq r$ iff $p \trianglelefteq_{\textup{homlin}} Q_{n,d}$, which is equivalent to
$p \in \overline{\GL_{nd} Q_{n,d}}$, provided $p$ is defined in the variables $x_{1,1},\ldots,x_{n,d}$.
In the algebraic complexity literature, Chow rank is called the homogeneous $\Sigma\Pi\Sigma$-circuit complexity.

The noncommutative analog (i.e., variables do not commute) of Chow rank is the classical tensor rank.
The notion of border rank for tensors was introduced in \cite{BCRL79} to construct faster-than-Strassen matrix multiplication algorithms. In \cite{Bin80}, Bini proved that tensor border rank and tensor rank define the same matrix multiplication exponent. Today this theory is deeply related to the study of Gorenstein algebras \cite{Iarrobino-Kanev,BuczBucz:SecantVarsHighDegVeroneseReembeddingsCataMatAndGorSchemes}, the Hilbert scheme of points \cite{Jeli:Pathologies}, and deformation theory \cite{BB-apolarity,JelMan:LimitsSaturatedIdeals}.
Homogeneous linear projections are used to compare not only the rank of tensors, but they are used to define a partial order on the set of all tensors, see e.g.~\cite[Ch~14.6]{burgisser2013algebraic}. This is also a common concept in quantum information theory.

\subsection{Homogeneous vs Inhomogeneous}
\label{subsec:homoginhomog}
In this subsection we work out the relation to classical (i.e., ungraded) algebraic complexity theory.
In order to define the notion of completeness of graded p-families for the classical algebraic complexity classes we use the following map $\varphi$.
Given $d \in \Pb$, $m\in \Pb$
and $a \in \IC^{\IN\times\IN}$,
then a graded p-family $f$ can be converted into an ungraded p-family $\varphi(f,a,m,d)$ by setting
$
\varphi(f,a,m,d)_n \ := \ \sum_{i=0}^{d_n} a_{n,i} \cdot f_{m_n,i}.
$
For a graded p-family $f$ we define the set $\varphi(f)$ of \emph{associated ungraded p-families} as
$\varphi(f) := \{ \varphi(f,a,m,d) \mid
m\in \Pb, \ d\in \Pb, \ a\in\IC^{\IN\times\IN}\}$.

\begin{definition}
\label{def:hardcomplete}
Let $\mathscr C \subseteq \IS^\IN$ be a class of ungraded p-families.
We say that a graded p-family $f$ is \emph{$\mathscr C$-hard}
if for all $g \in \mathscr C$ we have $g \leq_{\textup{p-homlin}}f$.

We say that $f$ is $\mathscr C$-complete
if $f$ is $\mathscr C$-hard and
$\varphi(f) \subseteq \mathscr C$.

\end{definition}
There are analogous variants for completeness under border projections ($g \trianglelefteq_{\textup{p-homlin}}f$) and quasipolynomial projections ($g \leq_{\textup{qp-homlin}}f$), and quasipolynomial border projections ($g \trianglelefteq_{\textup{qp-homlin}}f$).

The main example is that the graded p-family $\IMM$ is $\VBP$-complete under homogeneous linear p-projections.
From \S\oldref{subsec:classical}, $P$ is complete for the class of p-families with polynomially bounded Waring rank, and $Q$ is complete for the class of p-families with polynomially bounded Chow rank.

While for ungraded p-families we have to allow affine linear projections as reductions, for graded p-families we can (and always will) use the {\em weaker} 
notion of homogeneous linear projections.
Hence, it is {\em not obvious} how to turn a $\mathscr C$-complete ungraded $p$-family (under affine linear projections) into a $\mathscr C$-complete graded p-family (under homogeneous linear projections)! We illustrate this scenario by an example below.

Let us consider a ungraded p-family $g$, which is $\VF$-complete under affine linear projections; then $g$ interpreted as a graded p-family is {\em not necessarily} $\VF$-complete under homogeneous linear projections, as the following example illustrates.
The ungraded p-family $\IMM_3$ defined via $(\IMM_3)_n = \IMM_{3,n}$
is an ungraded $\VF$-complete p-family.
The constant ungraded p-family with each element $x_1^2+\cdots+x_{7}^2$ is in $\VF$, but by construction $\IMM_{3,2}$ is nonzero only for exactly $n=2$, and there is no homogeneous linear projection of $\IMM_{3,2}$ to
$x_1^2+\cdots+x_{7}^2$
(because every homogeneous linear projection of $\IMM_{3,2}$ has only at most 6 essential variables, i.e., its $\GL$-orbit has dimension at most 6.%
However, the reverse works under mild conditions on self-reducibility of $f$ under affine projections and on being able to simulate sums; as an example we refer to the following claim.
\begin{claim}We write $p\leq_{\textup{afflin}}q$ if $p$ can be obtained from $q$ by replacing variables in $q$ by affine linear polynomials.
Let $f$ be a graded p-family that is $\mathscr C$-complete under homogeneous linear projections, and assume that
$\forall n,d: f_{n,d-1}\leq_{\textup{afflin}} f_{n,d}$ and $f_{n-1,d}\leq_{\textup{homlin}} f_{n,d}$.
Let $g := \varphi(f,\textup{diag}(1,\ldots,1),\textup{id}_\IN,\textup{id}_\IN)$
with the property that
there exists a bivariately polynomially bounded $q$ such that $\forall n,k$: if $h_1,\ldots,h_k \leq_{\textup{afflin}} g_n$, then $h_1+\ldots+h_k \leq_{\textup{afflin}} g_{q(k,n)}$.
Then $g$ is $\mathscr C$-complete under affine linear projections.
\end{claim}
\begin{proof}
From $\varphi(f) \subseteq \mathscr C$ it follows that $g\in \mathscr C$. Now, let $h\in \mathscr C$ be an ungraded p-family.
We have $h \leq_{\textup{p-homlin}} f$, hence 
\[
\forall n,d : \ h_{n,d} \leq_{\textup{homlin}} f_{m_n,d} \,\leq_{\textup{afflin}}\, f_{\max\{m_n,\deg(h_n)\},\max\{m_n,\deg(h_n)\}} \;=\; g_{\max\{m_n,\deg(h_n)\}}.\]
Therefore, $\forall n: h_{n} \leq_{\textup{afflin}} g_{q(\deg(h_n)+1,\max\{m_n,\deg(h_n)\})}$.
Define
\[a(n) := q(\deg(h_n)+1,\max\{m_n,\deg(h_n)\}).\]
Thus, $\forall n: h_n\leq_{\textup{afflin}} g_{a(n)}$, which proves the claim, because $a\in \Pb$.
\end{proof}

While $\IMM$ is a $\VBP$-complete graded p-family
and $\IMM_3$ is a $\VF$-complete ungraded p-family,
our paper is the {\em first to introduce} a $\VF$-complete \emph{graded} p-family $\nc e_3$, see \ref{thm:nceVFpc}.
It is unclear if graded complete p-families for $\VP$ or for $\VNP$ \emph{exist}, and we leave this as an open question.
For example, it is not obvious if a universal circuit family for $\VP$ can be used to construct a $\VP$-complete graded p-family under homogeneous linear projections.

\section{Proof Ideas}

In this section, we briefly sketch the overall proof idea of \ref{thm:nceVFpc} and \ref{thm:intro:homogeneousVF}.
\subsection{Proof idea of \ref{thm:nceVFpc}}~The recent paper \cite[Section~3]{DGIJL23} introduced a notion of complexity with a rigid interplay between homogeneous linear entries and fixed constants, which they call Kumar's complexity. It is modeled after Kumar's construction in \cite{kum20}.
For a polynomial $f$, Kumar's complexity of~$f$ is the smallest
$m$ such that there exists a constant $\alpha$ and homogeneous linear polynomials $\ell_i$ such that
\begin{equation}\label{eq:Kc}
\textstyle
f = \alpha\big(\big(\prod_{i=1}^m (1+\ell_i)\big)-1\big).
\end{equation}

We study an analogous notion for matrices.
Let $E_{n,d}$ be the homogeneous degree $d$ part of the sum of the entries of
\[
\begin{pmatrix}
1 & x_{1,1,2} & x_{1,1,3} \\
x_{1,2,1} & 1 & x_{1,2,3} \\
x_{1,3,1} & x_{1,3,2} & 1
\end{pmatrix}
\cdots
\begin{pmatrix}
1 & x_{n,1,2} & x_{n,1,3} \\
x_{n,2,1} & 1 & x_{n,2,3} \\
x_{n,3,1} & x_{n,3,2} & 1
\end{pmatrix}
-
\begin{pmatrix}
1 & 0 & 0 \\
0 & 1 & 0 \\
0 & 0 & 1
\end{pmatrix}.
\]
In the expansion, the noncommutative elementary symmetric polynomials appear.
Our study of this setup leads to a homogenized version of the result by 
Ben-Or \& Cleve \cite{BC92}.
Here we have to pay close attention on how to deal with field constants, and we define the notion of \emph{input-homogeneous-linear computation} (IHL), see \S\oldref{subsec:inputhomcomp}. In particular, we prove an input-homogeneous-linear version of Brent's depth reduction, see \Cref{lem:inputhomogenize}.
\ref{thm:nceVFpc} appears in \S\oldref{sec:homogcompl} as \ref{cor:hardforVFH}.

\medskip
\subsection{Proof idea of \ref{thm:intro:homogeneousVF}} From $3 \times 3$ matrices, we turn to $2 \times 2$ matrices.
Note that (for odd $d$) $C_{n,d}$ is the homogeneous degree $d$ part of the $(1,2)$ entry of
$
\big(\begin{smallmatrix}
1 & x_1 \\
0 & 1
\end{smallmatrix}\big)
\big(\begin{smallmatrix}
1 & 0 \\
x_2 & 1
\end{smallmatrix}\big)
\cdots
\big(\begin{smallmatrix}
1 & x_n \\
0 & 1
\end{smallmatrix}\big)
-
\big(\begin{smallmatrix}
1 & 0 \\
0 & 1
\end{smallmatrix}\big).
$
\ref{thm:intro:homogeneousVF} appears in \S\oldref{sec:homogcompl} as \Cref{thm:hardness}.
Its proof is based on the construction of \cite{BIZ18}, which is, however inherently {\em affine}. To circumvent this, we convert the product gate into an arity 3 homogeneous product gate. The resulting analysis of arithmetic circuits and formulas allowing only arity 3 homogeneous product gates is surprisingly subtle. 
The graded p-family $C_{n,d}$ can be seen as a homogeneous variant of the continuant in \cite{BIZ18}.

For the last part of \ref{thm:intro:homogeneousVF}, we prove that $\VQTF = \VQP$.
The rest of the hardness proof follows then completely analogously via quasipolynomial homogeneous linear border projections.
The proof of $\VQTF = \VQP$ proceeds in two steps:
We first show that $\VF$ lies in 
$\VTP$ (the circuit analog of $\VTF$), see \ref{thm:VPVTP}, where we first ``parity-homogenize'' the formula (every gate has only even or only odd nonzero homogeneous components), and then compute $z \cdot f$ at each even-degree gate instead of $f$, where $z$ is a new variable. This additional factor $z$ is then later replaced, which is the main reason why the output of this construction is a {\em circuit} and not a formula.
Since we know that $\VTF\subseteq \VF$,
we are now in this situation:
\begin{align*}
    \VTF \ \subseteq \ \VF \ \subseteq \ \VTP \cap \VBP \ \subseteq \ \VP.
\end{align*}
Our proof does not give $\VTF=\VF$, see \Cref{rem:VFVTF}.
We conclude our proof by showing that $\VQTF = \VQTP$,
which implies that both classes are equal to $\VQTF = \VQF = \VQTP$, but we already know $\VQF=\VQP$.
For details, see \eqref{eq:inclusions} and~\ref{thm:VQTFVQTP}.

To achieve this, we use an arity-3 basis variant of the Valiant-Skyum-Berkowitz-Rackoff circuit depth reduction \cite{vsbr83}, which is a bit more involved than the original proof.

\section{Input-homogenization and Arity 3 Products}
\label{sec:homogcompl}

In this section, $f$ is a polynomial, and not a graded p-family.

\subsection{Input-homogeneous-linear Computation}
\label{subsec:inputhomcomp}

We start with a technicality in the definition of arithmetic circuits.
In this section every edge of an arithmetic circuit is labelled with a field constant. Instead of just forwarding the computation result of a gate to another gate, these edges rescale the polynomial along the way.
For arithmetic \emph{formulas} we do \emph{not} allow this, as we will see that it is unnecessary.
In other words, we allow $g+h$ gates for formulas, while we allow a $\alpha \cdot g + \beta \cdot h$ gates in circuits, and analogously for multiplication.

The \emph{depth} of an arithmetic circuit/formula is the length of the longest path from the source to a leaf.
\begin{definition}
An arithmetic formula/circuit is called \emph{input-homogeneous-linear} (IHL) if all its leaves are labelled with \emph{homogeneous linear} polynomials.
\end{definition}
In particular (contrary to ordinary arithmetic formulas/circuits)
in an IHL formula/circuit
we \emph{do not} allow any leaf to be labelled with a field constant. It now becomes clear why we needed the technicality:
For any $\alpha\in\IC$,
if an IHL circuit with $s$ gates computes a polynomial $f$, then using the scalars on the edges there exists an IHL circuit computing $\alpha f$ with also only $s$ many gates. For formulas this rescaling can be simulated by rescaling a subset of the leaves.
Indeed, we rescale the root of the formula by induction: we rescale a summation gate by rescaling both children, we rescale a product gate by rescaling an arbitrary child.
Alternatively, if $f$ is homogeneous, one can rescale the input gates by the $\sqrt[d]{\alpha}$. The latter technique works for formulas and circuits alike, but we will not use this method.

It is easy to see that IHL formulas/circuits can only compute polynomials $f$ with $f(0)=0$.
But other than that, being IHL is not a strong restriction, as the following simple lemma shows. We write $\widehat f := f-f(0)$.
\begin{lemma}
\label{lem:inputhomogenize}
Given an arithmetic circuit of size $s$ computing a polynomial $f$, then there exists an IHL arithmetic circuit of size $6s$ and depth $3s$ computing $\widehat f$.

There exists a polynomial $p$ such that:
Given any arithmetic formula of size $s$ computing a polynomial~$f$, then there exists an IHL arithmetic formula of size $p(s)$ and depth $O(\log(s))$ computing~$\widehat f$.
\end{lemma}
\begin{proof}
We treat the case of formulas first.
We first use Brent's depth reduction \cite{brent74} to ensure that the size is $\poly(s)$ and the depth is $O(\log(s))$.
We now proceed in a way that is similar to the homogenization of arithmetic circuits.
Let $F$ be the formula computing $f$. We replace every computation gate
(that computes some polynomial $g$)
by a pair of gates (and some auxiliary gates), one computing $g(0)$ and one computing $\widehat g$.
Clearly,

\begin{alignat*}{3}
& \big((g+h)(0),\,\widehat{g+h}\big) \;&&=\; \big(g(0)+h(0),\,\widehat g+\widehat h\big)\;\;\;\;&&(\textup{addition gate})\;,\\
     & \big((g \cdot h)(0),\,\widehat{g\cdot h}\big) \;&& =\; \big(g(0)\cdot h(0),\,g(0)\cdot \widehat h + \widehat g\cdot h(0) + \widehat g \cdot \widehat h \big)\;\;&&(\textup{multiplication gate})\;.
\end{alignat*}

Therefore, an addition gate is just replaced by 2 addition gates, while a multiplication gate is replaced by 4 multiplication gates and 2 addition gates (and this gadget has depth 3).
We copy the subformulas of $g(0)$, $h(0)$, $\widehat g$, and $\widehat h$, which maintains the depth, and it keeps the size $\poly(s)$.
In this construction additions happen only between constants or between non-constants, but never between a constant and a non-constant.
Therefore each maximal subformula of constant gates can be evaluated and replaced with a single constant gate, and these gates are multiplied with non-constant gates (with the one exception of the gate for $f(0)$).
But in a formula, scaling a non-constant gate by a field element {\em does not} require a multiplication gate, and instead we can recursively pass this scaling operation down to the children, as explained before this lemma.
At the end we remove the one remaining constant gate for $f(0)$ and are done.

For circuits we proceed similarly. We skip the depth reduction step.
Let $C$ be the formula computing $f$. We replace every computation gate
(that computes some polynomial $g$)
by a pair of gates (and some auxiliary gates), one computing $g(0)$ and one computing $\widehat g$.
Clearly, for addition and multiplication gates, we can do the following:
\begin{alignat*}{2}
& \big((\alpha g+\beta h)(0),\,\widehat{\alpha g+\beta h}\big) \;&&=\; \big(\alpha g(0)+\beta h(0),\,\alpha \widehat g+\beta \widehat h\big)\;,\\
     & \big((\alpha g \cdot \beta h)(0),\,\widehat{\alpha g\cdot \beta h}\big) \;&& =\; \big(\alpha g(0)\cdot \beta h(0),\,\alpha g(0)\cdot \beta \widehat h + \alpha \widehat g\cdot \beta h(0) + \alpha \widehat g \cdot \beta \widehat h \big)\;.
\end{alignat*}

Therefore, an addition gate is just replaced by 2 addition gates, while a multiplication gate is replaced by 4 multiplication gates and 2 addition gates (and this gadget has depth 3).
Here we have no need to copy subformulas, and we re-use the computation instead.
In this construction additions happen only between constants or between non-constants, but never between a constant and a non-constant.
Therefore each maximal subcircuit of constant gates can be evaluated and replaced with a single constant gate $v$, and each of these gates is multiplied with a non-constant gate $w$ (with the one exception of the gate for $f(0)$).
This rescaling of the polynomial computed at $w$ can be simulated by just rescaling all the edge labels of the outgoing edges of $w$, so $v$ can be removed.
At the end we also remove the one remaining constant gate for $f(0)$ and are done.
\end{proof}
A circuit/formula that is the sum of an IHL circuit/formula and a field constant is called an IHL$^+$ circuit/formula.
The following corollary is obvious.
\begin{corollary}\label{cor:VPVf}
$\VP$ is the set of p-families for which the IHL$^+$ circuit size is polynomially bounded.
$\VF$ is the set of p-families for which the IHL$^+$ formula size is polynomially bounded.
\end{corollary}
\begin{proof}
Use \ref{lem:inputhomogenize} to compute $\widehat f$.
The missing constant $f(0)$ can be added to the IHL circuit/formula as the very last operation.
\end{proof}

\subsection{IHL Ben-Or and Cleve is Exactly Kumar's complexity for \texorpdfstring{$3\times 3$}{3 x 3}  Matrices}

Quite surprisingly, the $3 \times 3$ matrix analogue of Kumar's complexity model (see \eqref{eq:Kc}) turns out to be the homogeneous version of Ben-Or and Cleve's construction \cite{BC92}, as the proof of the following \ref{pro:VFKumar} shows. Let $E_{i,j}$ denote the $3\times 3$ matrix with a $1$ at the entry $(i,j)$ and zeros elsewhere. Let $\idthree$ denote the $3 \times 3$ identity matrix.

\begin{proposition}
\label{pro:VFKumar}
Fix $i,j\in\{1,2,3\}$, $i\neq j$. Let $f$ be a polynomial admitting an IHL formula of depth~$\delta$. Then there exist $3 \times 3$ matrices $A_1 \vvirg A_r$ with $r \leq 4^\delta$ having \emph{homogeneous linear entries} such that
\[
f \cdot E_{i,j} \;=\; (\idthree+A_1)(\idthree+A_2)\cdots(\idthree+A_r) - \idthree\;.
\]
\end{proposition}
\begin{proof}
Consider the six positions $\{(i,j) \mid 1 \leq i,j \leq 3, \ i \neq j\}$ of the zeros in the $3\times 3$ unit matrix.
Given an IHL formula, to each input gate and to each computation gate we assign one of the 6 positions in the following way.
We start at the root and assign it position $(i,j)$.
We proceed by assigning position labels recursively:
For a summation gate with position $(i',j')$, both summands get position $(i',j')$. For a product gate with position $(i',j')$, one factor gets position $(i',k)$ and the other gets position $(k,j')$, $k\neq i'$, $k \neq j'$.
We now prove by induction on the depth $D$ of the gate $g$ (the depth of a gate it the depth of its subformula: the input have depth $0$; the root has the highest depth) with position $(i',j')$ that for each gate there is a list of at most $4^D$ matrices $(A_1,\ldots,A_{r})$ such that
\[
(\idthree+A_1)(\idthree+A_2)\cdots(\idthree+A_{r}) = \idthree + g E_{(i',j')}
\]
and the same number of matrices $B_1,\ldots,B_r$ such that
\[
(\idthree+B_1)(\idthree+B_2)\cdots(\idthree+B_{r}) = \idthree - g E_{(i',j')}.
\]
For an input gate (i.e., depth 0) with position $(i',j')$ and input label $\ell$, we set $A_1 := \ell \cdot E_{i',j'}$ and $B_1 := -\ell \cdot E_{i',j'}$.
For an addition gate with position $(i',j')$ let $(A_1,\ldots,A_r)$, $(B_1,\ldots,B_r)$ and $(A'_1,\ldots,A'_{r'})$, $(B'_1,\ldots,B'_{r'})$ be the lists coming from the induction hypothesis.
We define the list for the addition gate as the concatenations $(A_1,\ldots,A_r,A'_1,\ldots,A'_{r'})$
and
$(B_1,\ldots,B_r,B'_1,\ldots,B'_{r'})$.
Observe that
\[
(\idthree + f E_{(i',j')}) \cdot (\idthree + g E_{(i',j')}) \;=\; \idthree + (f+g)E_{(i',j')}\;,\]
and
\[(\idthree - f E_{(i',j')}) \cdot (\idthree - g E_{(i',j')}) \;=\; \idthree - (f+g)E_{(i',j')}\;.\]
so this case is correct.
For a product gate with position $(i',j')$
let $(A_1,\ldots,A_r)$, $(B_1,\ldots,B_r)$ and $(A'_1,\ldots,A'_{r'})$, $(B'_1,\ldots,B'_{r'})$ be the lists coming from the induction hypothesis,
i.e.,
$(\idthree+A_1)(\idthree+A_2)\cdots(\idthree+A_{r}) = \idthree + f E_{(i',k)}$,
$(\idthree+B_1)(\idthree+B_2)\cdots(\idthree+B_{r}) = \idthree - f E_{(i',k)}$,
$(\idthree+A'_1)(\idthree+A'_2)\cdots(\idthree+A'_{r}) = \idthree + g E_{(k,j')}$,
$(\idthree+B'_1)(\idthree+B'_2)\cdots(\idthree+B'_{r}) = \idthree - g E_{(k',j')}$.
Observe that
\[
\big(\idthree + f E_{(i',k)}\big)
\big(\idthree + g E_{(k,j')}\big)
\big(\idthree - f E_{(i',k)}\big)
\big(\idthree - g E_{(k,j')}\big)
=
\idthree + fg E_{(i',j')}
\]
and analogously
\[
\big(\idthree - f E_{(i',k)}\big)
\big(\idthree + g E_{(k,j')}\big)
\big(\idthree + f E_{(i',k)}\big)
\big(\idthree - g E_{(k,j')}\big)
=
\idthree - fg E_{(i',j')}.
\]
For illustration, in the notation of \cite{BIZ18} the product with position (1,3) can be depicted as follows.
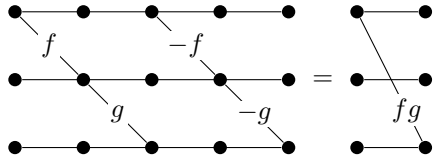
\begin{center}
\begin{tikzpicture}[scale=0.9]
\node[circle,fill=black,inner sep=0pt, minimum size = 0.5cm,minimum size=5pt] (00) at (0,0) {};
\node[circle,fill=black,inner sep=0pt, minimum size = 0.5cm,minimum size=5pt] (01) at (0,1) {};
\node[circle,fill=black,inner sep=0pt, minimum size = 0.5cm,minimum size=5pt] (02) at (0,2) {};
\node[circle,fill=black,inner sep=0pt, minimum size = 0.5cm,minimum size=5pt] (10) at (1,0) {};
\node[circle,fill=black,inner sep=0pt, minimum size = 0.5cm,minimum size=5pt] (11) at (1,1) {};
\node[circle,fill=black,inner sep=0pt, minimum size = 0.5cm,minimum size=5pt] (12) at (1,2) {};
\node[circle,fill=black,inner sep=0pt, minimum size = 0.5cm,minimum size=5pt] (20) at (2,0) {};
\node[circle,fill=black,inner sep=0pt, minimum size = 0.5cm,minimum size=5pt] (21) at (2,1) {};
\node[circle,fill=black,inner sep=0pt, minimum size = 0.5cm,minimum size=5pt] (22) at (2,2) {};
\node[circle,fill=black,inner sep=0pt, minimum size = 0.5cm,minimum size=5pt] (30) at (3,0) {};
\node[circle,fill=black,inner sep=0pt, minimum size = 0.5cm,minimum size=5pt] (31) at (3,1) {};
\node[circle,fill=black,inner sep=0pt, minimum size = 0.5cm,minimum size=5pt] (32) at (3,2) {};
\node[circle,fill=black,inner sep=0pt, minimum size = 0.5cm,minimum size=5pt] (40) at (4,0) {};
\node[circle,fill=black,inner sep=0pt, minimum size = 0.5cm,minimum size=5pt] (41) at (4,1) {};
\node[circle,fill=black,inner sep=0pt, minimum size = 0.5cm,minimum size=5pt] (42) at (4,2) {};
\node[circle,fill=black,inner sep=0pt, minimum size = 0.5cm,minimum size=5pt] (50) at (5,0) {};
\node[circle,fill=black,inner sep=0pt, minimum size = 0.5cm,minimum size=5pt] (51) at (5,1) {};
\node[circle,fill=black,inner sep=0pt, minimum size = 0.5cm,minimum size=5pt] (52) at (5,2) {};
\node[circle,fill=black,inner sep=0pt, minimum size = 0.5cm,minimum size=5pt] (60) at (6,0) {};
\node[circle,fill=black,inner sep=0pt, minimum size = 0.5cm,minimum size=5pt] (61) at (6,1) {};
\node[circle,fill=black,inner sep=0pt, minimum size = 0.5cm,minimum size=5pt] (62) at (6,2) {};
\node (c) at (4.5,1) {$=$};
\draw (00) -- (40);
\draw (01) -- (41);
\draw (02) -- (42);
\draw (50) -- (60);
\draw (51) -- (61);
\draw (52) -- (62);
\draw (02) -- (11) node[inner sep=1pt,midway,fill=white] {$f$};
\draw (11) -- (20) node[inner sep=1pt,midway,fill=white] {$g$};
\draw (22) -- (31) node[inner sep=1pt,midway,fill=white] {$-f$};
\draw (31) -- (40) node[inner sep=1pt,midway,fill=white] {$-g$};
\draw (52) -- (60) node[inner sep=1.5pt,near end,fill=white] {$fg$};
\end{tikzpicture}
\end{center}
Since $4 \cdot 4^{D-1} = 4^D$, the size bound is satisfied.
\end{proof}

Since the trace of a matrix can sometimes be preferrable to the $(i,j)$-entry,
we present the result with the trace, provided approximations are allowed.
\begin{proposition}
\label{pro:usethetrace}
For every IHL formula of depth $\delta$ there exist $\leq 4^\delta$ many
$3\times 3$ matrices $A_i$ with 
\emph{homogeneous linear entries} over $\IC[\eps,\eps^{-1}]$
and $\alpha \in \IC[\eps,\eps^{-1}]$
such that
\[
E_{1,1} \cdot f = \lim_{\eps\to0} \Big( \alpha \big((\idthree+A_1)(\idthree+A_2)\cdots(\idthree+A_r) - \idthree\big)\Big)
\]
and hence
\[
f = \lim_{\eps\to0} \textup{tr} \Big( \alpha \big((\idthree+A_1)(\idthree+A_2)\cdots(\idthree+A_r) - \idthree\big)\Big).
\]
\end{proposition}
\begin{proof}

The IHL formula is a sum of products of subformulas $g_1 \cdot h_1$, $g_2 \cdot h_2$, $\ldots$, $g_r \cdot h_r$, and $r\leq 2^\delta$ by induction.
We compute subformulas for $\eps g_1$, $-\eps g_1$, $\eps h_1$, $-\eps h_1$, $\eps g_2$, $-\eps g_2$, \ldots, $-\eps h_r$ as in the proof of \ref{pro:VFKumar} with position $(1,2)$ for each $\pm \eps g_i$ and position $(2,1)$ for each $\pm \eps h_i$.
It turns out that
\[
M_a := (\idthree+ \eps g_a E_{1,2})(\idthree+ \eps h_a E_{2,1})
(\idthree- \eps g_a E_{1,2})(\idthree- \eps h_a E_{2,1})
=
\idthree+ \eps^2 f_a g_a E_{1,1} + O(\eps^3).
\]
Pictorially:
\begin{center}
\begin{tikzpicture}[scale=0.95]
\node[circle,fill=black,inner sep=0pt, minimum size = 0.5cm,minimum size=5pt] (00) at (0,0) {};
\node[circle,fill=black,inner sep=0pt, minimum size = 0.5cm,minimum size=5pt] (01) at (0,1) {};
\node[circle,fill=black,inner sep=0pt, minimum size = 0.5cm,minimum size=5pt] (02) at (0,2) {};
\node[circle,fill=black,inner sep=0pt, minimum size = 0.5cm,minimum size=5pt] (10) at (1,0) {};
\node[circle,fill=black,inner sep=0pt, minimum size = 0.5cm,minimum size=5pt] (11) at (1,1) {};
\node[circle,fill=black,inner sep=0pt, minimum size = 0.5cm,minimum size=5pt] (12) at (1,2) {};
\node[circle,fill=black,inner sep=0pt, minimum size = 0.5cm,minimum size=5pt] (20) at (2,0) {};
\node[circle,fill=black,inner sep=0pt, minimum size = 0.5cm,minimum size=5pt] (21) at (2,1) {};
\node[circle,fill=black,inner sep=0pt, minimum size = 0.5cm,minimum size=5pt] (22) at (2,2) {};
\node[circle,fill=black,inner sep=0pt, minimum size = 0.5cm,minimum size=5pt] (30) at (3,0) {};
\node[circle,fill=black,inner sep=0pt, minimum size = 0.5cm,minimum size=5pt] (31) at (3,1) {};
\node[circle,fill=black,inner sep=0pt, minimum size = 0.5cm,minimum size=5pt] (32) at (3,2) {};
\node[circle,fill=black,inner sep=0pt, minimum size = 0.5cm,minimum size=5pt] (40) at (4,0) {};
\node[circle,fill=black,inner sep=0pt, minimum size = 0.5cm,minimum size=5pt] (41) at (4,1) {};
\node[circle,fill=black,inner sep=0pt, minimum size = 0.5cm,minimum size=5pt] (42) at (4,2) {};
\node[circle,fill=black,inner sep=0pt, minimum size = 0.5cm,minimum size=5pt] (50) at (5,0) {};
\node[circle,fill=black,inner sep=0pt, minimum size = 0.5cm,minimum size=5pt] (51) at (5,1) {};
\node[circle,fill=black,inner sep=0pt, minimum size = 0.5cm,minimum size=5pt] (52) at (5,2) {};
\node[circle,fill=black,inner sep=0pt, minimum size = 0.5cm,minimum size=5pt] (60) at (6,0) {};
\node[circle,fill=black,inner sep=0pt, minimum size = 0.5cm,minimum size=5pt] (61) at (6,1) {};
\node[circle,fill=black,inner sep=0pt, minimum size = 0.5cm,minimum size=5pt] (62) at (6,2) {};
\node (c) at (4.5,1) {$=$};
\node (eps) at (7,1) {$+O(\eps^3)$};
\draw (00) -- (40);
\draw (01) -- (41);
\draw (02) -- (42);
\draw (50) -- (60);
\draw (51) -- (61);
\draw (52) -- (62) node[inner sep=1.5pt,midway,above] {\scalebox{0.75}{$1+\eps^2 h_a g_a$}};
\draw (02) -- (11) node[inner sep=1.5pt,midway,fill=white] {$\eps g_a$};
\draw (11) -- (22) node[inner sep=1.5pt,midway,fill=white] {$\eps h_a$};
\draw (22) -- (31) node[inner sep=1.5pt,midway,fill=white] {$-\eps g_a$};
\draw (31) -- (42) node[inner sep=1.5pt,midway,fill=white] {$-\eps h_a$};
\end{tikzpicture}
\end{center}
Hence $M_1 M_2 \cdots M_r = \idthree + \eps^2(h_1g_1+h_2g_2+\cdots+h_rg_r)E_{1,1} + O(\eps^3)$.
We choose $\alpha = \eps^{-2}$.
\end{proof}

Recall, $\nc e_{n,d}(X_1, \dots, X_n) \;:=\; \sum_{1 \leq I_1 < I_2 <\dots < I_d \leq n} X_{I_1} \dots X_{I_d}$,
is the elementary symmetric polynomial in noncommuting variables $X_1,\ldots,X_n$.
For any $L : \IC^{3\times 3}\to \IC$,
let $\nc e_{L,n,d} := L \circ \nc e_d(A_1,A_2,\ldots,A_n)$, where each
$A_i=\left(\begin{smallmatrix}
0&x_{1,2,i}&x_{1,3,i}
\\
x_{2,1,i}&0&x_{2,3,i}
\\
x_{3,1,i}&x_{3,2,i}&0
\end{smallmatrix}\right)$
is a $3\times 3$ matrix of 6 fresh variables.
We denote by $\nc e_L$ the corresponding graded p-family.
To be formally precise, we set $\nc e_{L,n,0}=1$.
In particular, $L$ can be taken to be the trace.

\begin{corollary}
\label{cor:hardforVFH}
Fix any nonzero linear form $L$ on the space of $3 \times 3$ matrices.
If $L$ is supported outside the main diagonal, then
the graded p-family $\nc e_L$ is $\VF$-complete under homogeneous linear projections.
If $L$ is supported on the main diagonal, then
the graded p-family $\nc e_L$ is $\VF$-complete under homogeneous linear border projections.

\end{corollary}
\begin{proof}
Given a ungraded p-family $g\in\VF$.
We apply Brent's depth reduction and then \ref{lem:inputhomogenize} to every homogeneous component of every $g_n$ to obtain IHL formulas $f_{n,d}$ of logarithmic depth and polynomial size in~$n$ ($d$ is polynomial in $n$).
The first case is treated with \ref{pro:VFKumar}, the second is treated completely analogously with \ref{pro:usethetrace}.
We only handle the slightly more difficult second case.
We obtain $4^{O(\log n)}  =\poly(n)$ many matrices $A_i$ with
\[
f_n = \lim_{\eps\to0} L\Big( \alpha \big((\idthree+A_1)(\idthree+A_2)\cdots(\idthree+A_r) - \idthree\big)\Big)
\]
Note that $\alpha\in\IC[\eps,\eps^{-1}]$ can be assumed to be a scalar times a power of $\eps$, because lower order terms have no effect on the limit.
Since $f_{n,d}$ is homogeneous of degree $d$, we have
\[
f_{n,d} = \lim_{\eps\to0} L\Big( \alpha \, \nc e_{n,d}(A_1,\ldots,A_r)\Big)
= \lim_{\eps\to0} L\Big(  \nc e_{n,d}(\sqrt[d]{\beta}\eps^k A'_1,\ldots,\alpha \sqrt[d]{\beta}\eps^k A'_r)\Big)
\]
where $A'_i$ arises from $A_i$ by replacing every $\eps$ by $\eps^d$.
\end{proof}

While \ref{cor:hardforVFH} gives the first collection that is $\VF$-complete under homogeneous linear projection, we found simpler polynomials with similar properties.
In the next sections we will prove that the parity-alternating elementary symmetric polynomial is
hard for the class $\VTF$ under homogeneous linear projections, which gives a polynomial that is just barely more complicated than the elementary symmetric polynomial.

\subsection{IHL Computation with Arity 3 Products}
\label{subsec:aritythree}
In the light of \cite{BIZ18} we now study the $2 \times 2$ analogues of \ref{pro:VFKumar}, \ref{pro:usethetrace}, \ref{cor:hardforVFH}.
In order to do so, in this section we study IHL formulas and circuits where the additions have arity 2, but the products have \emph{arity exactly 3}.
We call this basis the \emph{arity 3 basis}.
This turns out to be rather subtle, because one would usually want to simulate an arity 2 product by an arity 3 product in which one of the factors is a constant 1, but that violates the IHL property.
A circuit/formula of this type is called an \emph{IHL circuit/formula over the arity 3 basis}.
If a polynomial is computed by an IHL formula or circuit over the arity 3 basis, then all its homogeneous even-degree components are zero, hence we have to adjust this definition slightly: For an even degree homogeneous polynomial we want to compute all partial derivatives instead.
Formally, a \emph{graded IHL circuit/formula over the arity 3 basis} is a circuit/formula of the following syntactic structure:
\[
f \ = \ f(0) \ + \ \sum_{d\in 2\IN+1} \underbrace{f_{d}}_{\textup{IHL, arity 3}}
 \ + \ 
\sum_{\substack{d \in 2\IN\\d\geq 2}}
\tfrac{1}{d}
\sum_{i=1}^m x_i \cdot
\underbrace{\partial f_d/\partial x_i}_{\textup{IHL, arity 3}},
\]
where each homogeneous $f_d$ and homogeneous $\partial f_d/\partial x_i$ is computed by an IHL circuit/formula over the arity 3 basis.
Euler's homogeneous function theorem ensures that the right-hand side actually computes $f$.
We define $\VTP$ and $\VTF$ as follows:
\begin{definition}[$\VTP$ and $\VTF$] \label{def:v3f}
$\VTP$ (respectively,~$\VTF$) is the class of p-families for which
the graded IHL circuit (respectively,~formula) complexity over the arity 3 basis is polynomially bounded.
\end{definition}
We have the following inclusion among the classes:
\begin{align}
\label{eq:inclusions}
\VTF \ \subseteq \ \VF \ \subseteq \ \VTP \cap \VBP \ \subseteq \ \VP,
\end{align}
where $\VTF\subseteq\VF$ is obvious, and we prove the first inclusion in \ref{thm:VPVTP}, while it is well-known that $\VF \subseteq \VBP$.
It is known that if we go to quasipolynomial complexity instead of polynomial complexity, the three classical classes coincide: $\VQF = \VQBP = \VQP$, which is an immediate corollary of the circuit depth reduction result of Valiant-Berkowitz-Skyum-Rackoff \cite{vsbr83}.
We prove in \ref{thm:VQTFVQTP} that our two new classes also belong to this set: All classes in \eqref{eq:inclusions} coincide if we go to quasipolynomial complexity instead of polynomial complexity, see \eqref{eq:quasipolyallfive}.

The following proposition is an adaption of Brent's depth reduction \cite{brent74}
and it shows that instead of polynomially sized formulas we can work with formulas of logarithmic depth. Both properties, IHL and the arity 3 basis, require some moficiations to Brent's original argument.
\begin{proposition}[Brent's depth reduction for graded IHL formulas over the arity 3 basis]\label{pro:arityIIIbrent}
Let $f$ be a polynomial computed by a graded IHL formula of size $s$ over the arity 3 basis. Then there exists a graded IHL formula over the arity 3 basis of size $\poly(s)$ and depth $O(\log(s))$ computing $f$.
\end{proposition}
\begin{proof}
We discuss only the homogeneous odd-degree case, because the more general case directly follows from it.
The construction is recursive, just as in Brent's original argument.
We follow the description in \cite{Sap19}.
We start at the root and keep picking the child with the larger subformula until we reach a vertex $v$ with $\frac 1 3 s \leq |\langle v\rangle|\leq \frac 2 3 s$, where $\langle v \rangle$ is the subformula at the gate $v$.
We make a case distinction.
In the first case we assume that on the path from $v$ to the root (excluding $v$) there is no product gate.
We reorder the gates as follows:
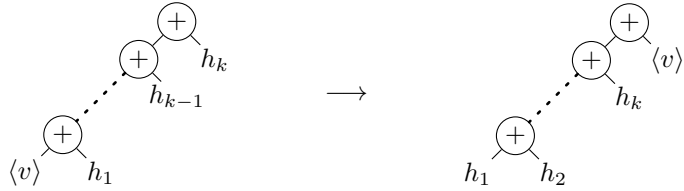
\[
\begin{tikzpicture}
\node at (0,0) {
\begin{tikzpicture}[scale=0.5]
\node[draw,circle,inner sep=0pt, minimum size = 0.5cm] (1) at (0,0) {+};
\node[draw,circle,inner sep=0pt, minimum size = 0.5cm] (2) at (-1,-1) {+};
\node[draw,circle,inner sep=0pt, minimum size = 0.5cm] (4) at (-3,-3) {+};
\node[inner sep=1pt] (hk) at (1,-1) {$h_k$};
\node[inner sep=1pt] (hkm1) at (0,-2) {$h_{k-1}$};
\node[inner sep=1pt] (h1) at (-2,-4) {$h_1$};
\node[inner sep=1pt] (v) at (-4,-4) {$\langle v\rangle$};
\draw (1) -- (2);
\draw[line width=1pt, line cap = round, loosely dotted] (2) -- (4);
\draw (1) -- (hk);
\draw (2) -- (hkm1);
\draw (4) -- (h1);
\draw (4) -- (v);
\end{tikzpicture}
};
\node at (3,0) {
$\longrightarrow$
};
\node at (6,0) {
\begin{tikzpicture}[scale=0.5]
\node[draw,circle,inner sep=0pt, minimum size = 0.5cm] (1) at (0,0) {+};
\node[draw,circle,inner sep=0pt, minimum size = 0.5cm] (2) at (-1,-1) {+};
\node[draw,circle,inner sep=0pt, minimum size = 0.5cm] (4) at (-3,-3) {+};
\node[inner sep=1pt] (hk) at (1,-1) {$\langle v\rangle$};
\node[inner sep=1pt] (hkm1) at (0,-2) {$h_k$};
\node[inner sep=1pt] (h1) at (-2,-4) {$h_2$};
\node[inner sep=1pt] (v) at (-4,-4) {$h_1$};
\draw (1) -- (2);
\draw[line width=1pt, line cap = round, loosely dotted] (2) -- (4);
\draw (1) -- (hk);
\draw (2) -- (hkm1);
\draw (4) -- (h1);
\draw (4) -- (v);
\end{tikzpicture}
};
\end{tikzpicture}
\]
The construction applied to a size $s$ formula gives $\textup{Depth}(s)\leq \textup{Depth}(\frac 2 3 s)+1$.
The resulting size is $\textup{Size}(s) \leq 2 \cdot \textup{Size}(\frac 2 3 s)+1$.

In the second case we assume that $v$ is the child of a product gate.
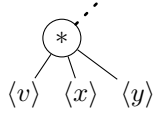
\[
\begin{tikzpicture}[scale=0.5]
\node[draw,circle,inner sep=0pt, minimum size = 0.5cm] (p2) at (6,1.5) {$\ast$};
\node[inner sep=1pt] (vnew) at (5,0) {$\langle v\rangle$};
\node[inner sep=1pt] (xnew) at (6.5,0) {$\langle x\rangle$};
\node[inner sep=1pt] (ynew) at (8,0) {$\langle y\rangle$};
\draw (p2) -- (vnew);
\draw (p2) -- (xnew);
\draw (p2) -- (ynew);
\draw[line width=1pt, line cap = round, loosely dotted] (p2) -- (7,2.5);
\end{tikzpicture}
\]
We now replace $\langle v\rangle$ by a new variable $\alpha$ and $\langle x\rangle$ by a new variable $\beta$.
We observe that the resulting polynomial $F$ (interpreted as a bivariate polynomial in $\alpha$ and $\beta$) is \emph{linear} in the product $\alpha\beta$. Therefore $F(\alpha,\beta) = \alpha\beta (F(1,1) - F(0,0)) + F(0,0)$.
Both $F(0,0)$ and $F(1,1)$ can be realized as an IHL formula over the arity 3 basis (because an arity 3 product gate with two 1s as inputs can be replaced by just the third input, and an arity 3 product gate with two 0s as input can be replaced by a constant 0, which can be simulated by removing gates), so we obtain:
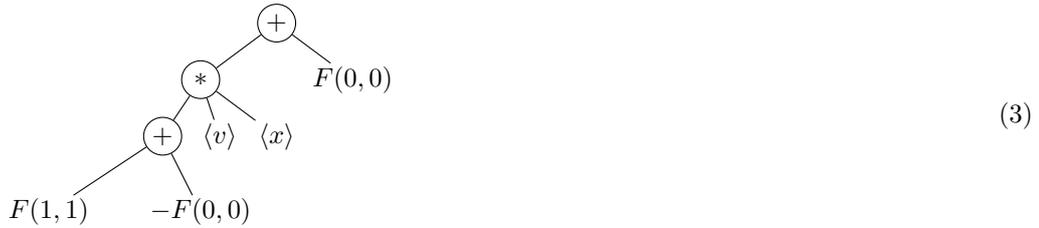
\begin{equation}\label{eq:Brent}
\begin{minipage}{8cm}
\begin{tikzpicture}[scale=0.5]
\node[draw,circle,inner sep=0pt, minimum size = 0.5cm] (top) at (8,3) {+};
\node[draw,circle,inner sep=0pt, minimum size = 0.5cm] (p2) at (6,1.5) {$\ast$};
\node[draw,circle,inner sep=0pt, minimum size = 0.5cm] (vnew) at (5,0) {$+$};
\node[inner sep=1pt] (F11) at (2,-2) {$F(1,1)$};
\node[inner sep=1pt] (F00) at (6,-2) {$-F(0,0)$};
\node[inner sep=1pt] (F00r) at (10,1.5) {$F(0,0)$};
\node[inner sep=1pt] (xnew) at (6.5,0) {$\langle v\rangle$};
\node[inner sep=1pt] (ynew) at (8,0) {$\langle x\rangle$};
\draw (top) -- (F00r);
\draw (top) -- (p2);
\draw (vnew) -- (F00);
\draw (vnew) -- (F11);
\draw (p2) -- (vnew);
\draw (p2) -- (xnew);
\draw (p2) -- (ynew);
\end{tikzpicture}
\end{minipage}
\end{equation}
The construction on a size $s$ formula gives $\textup{Depth}(s)\leq \textup{Depth}(\frac 2 3 s) + 2$.
The resulting size is: $\textup{Size}(s) \leq 5 \cdot \textup{Size}(\frac 2 3 s)+3$.

In the third case we assume that on the path from from $v$ to the root (excluding $v$)
there are addition gates and then a product gate, so
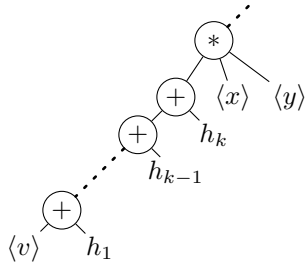
\[
\begin{tikzpicture}[scale=0.5]
\node[draw,circle,inner sep=0pt, minimum size = 0.5cm] (p) at (1,1.5) {$\ast$};
\node[draw,circle,inner sep=0pt, minimum size = 0.5cm] (1) at (0,0) {+};
\node[draw,circle,inner sep=0pt, minimum size = 0.5cm] (2) at (-1,-1) {+};
\node[draw,circle,inner sep=0pt, minimum size = 0.5cm] (4) at (-3,-3) {+};
\node[inner sep=1pt] (hk) at (1,-1) {$h_k$};
\node[inner sep=1pt] (hkm1) at (0,-2) {$h_{k-1}$};
\node[inner sep=1pt] (h1) at (-2,-4) {$h_1$};
\node[inner sep=1pt] (v) at (-4,-4) {$\langle v\rangle$};
\node[inner sep=1pt] (x) at (1.5,0) {$\langle x\rangle$};
\node[inner sep=1pt] (y) at (3,0) {$\langle y\rangle$};
\draw[line width=1pt, line cap = round, loosely dotted] (p) -- (2,2.5);
\draw (p) -- (1);
\draw (p) -- (x);
\draw (p) -- (y);
\draw (1) -- (2);
\draw[line width=1pt, line cap = round, loosely dotted] (2) -- (4);
\draw (1) -- (hk);
\draw (2) -- (hkm1);
\draw (4) -- (h1);
\draw (4) -- (v);
\end{tikzpicture}
\]
As a first step we make copies of $\langle x\rangle$ and $\langle y\rangle$ and call them $\langle x'\rangle$ and $\langle y'\rangle$, respectively, and re-wire similarly as in the first case:
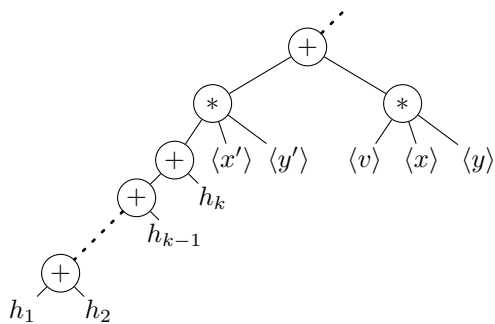
\[
\begin{tikzpicture}[scale=0.5]
\node[draw,circle,inner sep=0pt, minimum size = 0.5cm] (p) at (1,1.5) {$\ast$};
\node[draw,circle,inner sep=0pt, minimum size = 0.5cm] (p2) at (6,1.5) {$\ast$};
\node[draw,circle,inner sep=0pt, minimum size = 0.5cm] (1) at (0,0) {+};
\node[draw,circle,inner sep=0pt, minimum size = 0.5cm] (2) at (-1,-1) {+};
\node[draw,circle,inner sep=0pt, minimum size = 0.5cm] (4) at (-3,-3) {+};
\node[draw,circle,inner sep=0pt, minimum size = 0.5cm] (top) at (3.5,3) {+};
\node[inner sep=1pt] (hk) at (1,-1) {$h_k$};
\node[inner sep=1pt] (hkm1) at (0,-2) {$h_{k-1}$};
\node[inner sep=1pt] (h1) at (-2,-4) {$h_2$};
\node[inner sep=1pt] (v) at (-4,-4) {$h_1$};
\node[inner sep=1pt] (x) at (1.5,0) {$\langle x'\rangle$};
\node[inner sep=1pt] (y) at (3,0) {$\langle y'\rangle$};
\node[inner sep=1pt] (vnew) at (5,0) {$\langle v\rangle$};
\node[inner sep=1pt] (xnew) at (6.5,0) {$\langle x\rangle$};
\node[inner sep=1pt] (ynew) at (8,0) {$\langle y\rangle$};
\draw[line width=1pt, line cap = round, loosely dotted] (top) -- (4.5,4);
\draw (p2) -- (vnew);
\draw (p2) -- (xnew);
\draw (p2) -- (ynew);
\draw (top) -- (p);
\draw (top) -- (p2);
\draw (p) -- (1);
\draw (p) -- (x);
\draw (p) -- (y);
\draw (1) -- (2);
\draw[line width=1pt, line cap = round, loosely dotted] (2) -- (4);
\draw (1) -- (hk);
\draw (2) -- (hkm1);
\draw (4) -- (h1);
\draw (4) -- (v);
\end{tikzpicture}
\]
On the right-hand side of the tree we now proceed analogously as in the second case.
We replace $\langle v\rangle$ by a new variable $\alpha$ and $\langle x\rangle$ by a new variable $\beta$.
We observe that the resulting polynomial $F$ (interpreted as a bivariate polynomial in $\alpha$ and $\beta$) is linear in the product $\alpha\beta$. Therefore, 
\[
F(\alpha,\beta) \;=\; \alpha\beta (F(1,1) - F(0,0)) + F(0,0)\;.
\]
Both $F(0,0)$ and $F(1,1)$ can be realized as an input-homogeneous formula over the arity 3 basis, so we obtain the same formula as in \eqref{eq:Brent}.
The construction on a size $s$ formula gives $\textup{Depth}(s)\leq \textup{Depth}(\frac 2 3 s) + 2$.
The resulting size is $\textup{Size}(s) \leq 5 \cdot \textup{Size}(\frac 2 3 s)+3$.
Putting all cases together, the construction has $\textup{Depth}(s)\leq \textup{Depth}(\frac 2 3 s) + 2$ and $\textup{Size}(s) \leq 5 \cdot \textup{Size}(\frac 2 3 s)+3$.
Hence applying the construction recursively gives logarithmic depth and polynomial size.
\end{proof}

\subsection{The Parity-alternating Elementary Symmetric Polynomial}
\label{sec:parityalternating}

Let $n$ be odd. For odd $i$ let $X_i = \big(\begin{smallmatrix}0 & x_i \\ 0 & 0\end{smallmatrix}\big)$, and for even $i$ let
$X_i = \big(\begin{smallmatrix}0 & 0 \\ x_i & 0\end{smallmatrix}\big)$.
Let $A := \nc e_{n,d}(X_1,X_2,\ldots,X_n)$.
Note that in row 1 the matrix $A$ has only one nonzero entry, and its position depends on the parity of~$n$.
Let $C_{n,d} := A_{1,1}+A_{1,2}$.
A sequence $a$ of integers is called \emph{parity-alternating} if $a_i \neq a_{i+1} \mod 2$ for all $i$, and $a_1$ is odd.
Let $P$ be the set of length $d$ increasing parity-alternating sequences of numbers from $\{1,\ldots,n\}$.
It is easy to see that
\begin{equation}
\label{eq:Cnd}\textstyle
C_{n,d} = \sum_{(i_1,i_2,\ldots,i_d) \in P} \, x_{i_1}x_{i_2} \cdots x_{i_d}.
\end{equation}
We call the corresponding graded p-family $C$.
We usually only consider the case when the parities of $d$ and $n$ coincide, which is justified by the following lemma.
\begin{lemma}
If $n$ and $d$ have different parity, then
$C_{n,d} = C_{n-1,d}$.
\end{lemma}
\begin{proof}
If $d$ is odd, each parity-alternating sequence always ends with an odd parity, so if $n$ is even we have $C_{n,d}=C_{n-1,d}$.
If $d$ is even, each parity-alternating sequence always ends with an even parity, so if $n$ is odd we have $C_{n,d}=C_{n-1,d}$.
\end{proof}
Analogously to \ref{cor:hardforVFH} we have the following theorem.
\begin{theorem}
\label{thm:hardness}
Recall $C$ from \eqref{eq:Cnd} and $\varphi$ from \S\oldref{subsec:homoginhomog}.
The graded p-family $C$ is $\VTF$-hard under homogeneous linear border projections,
and $\varphi(C)\subseteq\VF$.

\end{theorem}

\begin{proof}
Let $\idtwo$ denote the $2 \times 2$ identity matrix.
$\varphi(C)\subseteq \VF$ follows from the fast that $C_{n,d}$ is the homogeneous degree $d$ component of the product $(\textup{id}_2+X_1)\cdots(\textup{id}_2+X_n)$ of $2\times 2$ matrices.

We prove $\VTF$-hardness.
Given a logdepth formula for a homogeneous degree $d$ polynomial $f$.
Let $E_{\textup{odd}} = \big(\begin{smallmatrix}
0&1\\
0&0
\end{smallmatrix}\big)
$
and let
$E_{\textup{even}} = \big(\begin{smallmatrix}
0&0\\
1&0
\end{smallmatrix}\big)
$.
We are given a formula for a homogeneous degree $d$ polynomial~$f$.
We can assume that the gates are additions and negative cubes ($x \mapsto -x^3$), because
$xyz=\frac{1}{24}\big(
(x+y+z)^3-(x+y-z)^3-(x-y+z)^3+(x-y-z)^3
\big)$, and the rescalings by $(\pm 24)^{-\frac 1 3}$ can be pushed to the input gates.
We first treat the case of $d$ being odd.
We write $A \simeq B$ is $A$ and $B$ are parametrized by $\eps$ and both limits $\lim_{\eps\to 0}A$ and $\lim_{\eps\to 0}B$ exist and coincide with each other.
We prove by induction on the depth $D$ of a gate that
there exist $\leq 3^D$ homogeneous linear forms $\ell_1,\ldots,\ell_r$ over $\IC[\eps,\eps^{-1},\alpha]$
such that
\[
\alpha f \cdot E_{\textup{odd}} 
\simeq
(\idtwo+\ell_1 E_{\textup{odd}})
(\idtwo+\ell_2 E_{\textup{even}})
\cdots
(\idtwo+\ell_r E_{\textup{odd}})
-\idtwo %
\]

The induction starting at an input gate with label $\ell$
is done by $\ell_1 = \alpha\ell$.
The addition gate is handled as follows.
By induction hypothesis there exist $\ell_1,\ldots,\ell_r$ and $\ell'_1,\ldots,\ell'_{r'}$
with
\[
\alpha f \cdot E_{\textup{odd}}
+ \idtwo
\simeq (\idtwo+\ell_1 E_{\textup{odd}})
(\idtwo+\ell_2 E_{\textup{even}})
\cdots
(\idtwo+\ell_r E_{\textup{odd}})
\qquad \textup{and}
\]
\[
\alpha g \cdot E_{\textup{odd}}
+ \idtwo
\simeq
(\idtwo+\ell'_1 E_{\textup{odd}})
(\idtwo+\ell'_2 E_{\textup{even}})
\cdots
(\idtwo+\ell'_{r'} E_{\textup{odd}})
\]
Therefore $\alpha (f+g) \cdot E_{\textup{odd}}
+ \idtwo = 
(\alpha f \cdot E_{\textup{odd}}
+ \idtwo)
(\alpha g \cdot E_{\textup{odd}}
+ \idtwo) \simeq
$
\[
(\idtwo+\ell_1 E_{\textup{odd}})
(\idtwo+\ell_2 E_{\textup{even}})
\cdots
(\idtwo+\ell_r E_{\textup{odd}})
(\idtwo+\ell'_1 E_{\textup{odd}})
(\idtwo+\ell'_2 E_{\textup{even}})
\cdots
(\idtwo+\ell'_{r'} E_{\textup{odd}})
\]

Handling the negative cube gates is more subtle (the negative squaring gates are also the subtle cases in \cite{BIZ18}).
By induction hypothesis we have $\ell_1,\ldots,\ell_r$ such that
\begin{equation}\label{eq:indhyp}
\alpha f \cdot E_{\textup{odd}} \simeq (\idtwo+\ell_1 E_{\textup{odd}})
(\idtwo+\ell_2 E_{\textup{even}})
\cdots
(\idtwo+\ell_r E_{\textup{odd}})
-\idtwo %
\end{equation}
We replace each $\eps$ by $\eps^k$ in each $\ell_i$, with $k$ so large that even when we replace $\alpha$ by $\eps^{-1}$ or $-\eps^{-1}$, we still have the equivalence of the LHS and RHS mod $\eps^2$.

We call the resulting linear forms $\ell_i'$. It follows that
\[
\alpha f \cdot E_{\textup{odd}} \equiv \big((\idtwo+\ell'_1 E_{\textup{odd}})
(\idtwo+\ell'_2 E_{\textup{even}})
\cdots
(\idtwo+\ell'_r E_{\textup{odd}})
-\idtwo\big)  \pmod{\eps^k}
\]
Setting $\alpha$ to $\eps^{-1}$ we obtain
\[
\eps^{-1} f \cdot E_{\textup{odd}} \equiv \big((\idtwo+\ell''_1 E_{\textup{odd}})
(\idtwo+\ell''_2 E_{\textup{even}})
\cdots
(\idtwo+\ell''_r E_{\textup{odd}})
-\idtwo\big)  \pmod{\eps^2}
\]
Anaogously with $\alpha = -\eps^{-1}$:
\[
-\eps^{-1} f \cdot E_{\textup{odd}} \equiv \big((\idtwo+\tilde\ell''_1 E_{\textup{odd}})
(\idtwo+\tilde\ell''_2 E_{\textup{even}})
\cdots
(\idtwo+\tilde\ell''_r E_{\textup{odd}})
-\idtwo\big)  \pmod{\eps^2}
\]
The induction hypothesis \eqref{eq:indhyp} also implies (set $\eps$ to $\eps^3$ and $\alpha$ to $\eps^2\alpha$) that
\[
\eps^2 \alpha f \cdot E_{\textup{odd}} \equiv \big((\idtwo+\ell'''_1 E_{\textup{odd}})
(\idtwo+\ell'''_2 E_{\textup{even}})
\cdots
(\idtwo+\ell'''_r E_{\textup{odd}})
-\idtwo\big) \pmod{\eps^3}
\]
Transposing gives
\[
\eps^2 \alpha f \cdot E_{\textup{even}} \equiv \big((\idtwo+\ell'''_r E_{\textup{even}})
(\idtwo+\ell'''_{r-1} E_{\textup{odd}})
\cdots
(\idtwo+\ell'''_1 E_{\textup{even}})
-\idtwo\big) \pmod{\eps^3}
\]
We now observe:
\[
(\eps^{-1}f E_{\textup{odd}} + \idtwo + \eps^2 g_1)
(\eps^{2}\alpha f E_{\textup{even}} + \idtwo  + \eps^3 g_2)
(-\eps^{-1}f E_{\textup{odd}} + \idtwo  + \eps^2 g_3)
\simeq
-\alpha f^3 E_{\textup{odd}} + \idtwo.
\]

\begin{center}
\begin{tikzpicture}[scale=1.2]
\node at (-3,0.5) {Pictorially:};
\node[circle,fill=black,inner sep=0pt, minimum size = 0.5cm,minimum size=5pt] (00) at (0,0) {};
\node[circle,fill=black,inner sep=0pt, minimum size = 0.5cm,minimum size=5pt] (01) at (0,1) {};
\node[circle,fill=black,inner sep=0pt, minimum size = 0.5cm,minimum size=5pt] (10) at (1,0) {};
\node[circle,fill=black,inner sep=0pt, minimum size = 0.5cm,minimum size=5pt] (11) at (1,1) {};
\node[circle,fill=black,inner sep=0pt, minimum size = 0.5cm,minimum size=5pt] (20) at (2,0) {};
\node[circle,fill=black,inner sep=0pt, minimum size = 0.5cm,minimum size=5pt] (21) at (2,1) {};
\node[circle,fill=black,inner sep=0pt, minimum size = 0.5cm,minimum size=5pt] (30) at (3,0) {};
\node[circle,fill=black,inner sep=0pt, minimum size = 0.5cm,minimum size=5pt] (31) at (3,1) {};
\node at (0.5,-0.2) {\scalebox{0.75}{$+O(\eps^2)$}};
\node at (1.5,-0.2) {\scalebox{0.75}{$+O(\eps^3)$}};
\node at (2.5,-0.2) {\scalebox{0.75}{$+O(\eps^2)$}};
\node at (3.5,0.5) {$=$};
\node[circle,fill=black,inner sep=0pt, minimum size = 0.5cm,minimum size=5pt] (40) at (4,0) {};
\node[circle,fill=black,inner sep=0pt, minimum size = 0.5cm,minimum size=5pt] (41) at (4,1) {};
\node[circle,fill=black,inner sep=0pt, minimum size = 0.5cm,minimum size=5pt] (50) at (5,0) {};
\node[circle,fill=black,inner sep=0pt, minimum size = 0.5cm,minimum size=5pt] (51) at (5,1) {};
\draw (00) -- (30);
\draw (01) -- (31);
\draw (01) -- (10) node[midway,fill=white] {$\varepsilon^{-1}f$};
\draw (10) -- (21) node[midway,fill=white] {$\varepsilon^{2}\alpha f$};
\draw (21) -- (30) node[midway,fill=white] {$-\varepsilon^{-1}f$};
\draw (40) -- (50);
\draw (41) -- (51);
\draw (41) -- (50) node[midway,fill=white] {$-\alpha f^3$};
\node at (4.5,-0.2) {\scalebox{0.75}{$+O(\eps)$}};
\end{tikzpicture}
\end{center}
At the end, setting $\alpha=1$ we obtain
\[
\alpha f \cdot E_{\textup{odd}} \;\simeq\; (\idtwo+\ell_1 E_{\textup{odd}})
(\idtwo+\ell_2 E_{\textup{even}})
\cdots
(\idtwo+\ell_r E_{\textup{odd}})
-\idtwo.
\]
Observe that $r$ is only polynomially large, because we started with a formula of logarithmic depth.
Since $f$ is homogeneous of degree $d$, this implies
\[
f \ \simeq \ \nc e_{r,d}(\ell_1 E_{\textup{odd}},\ell_2 E_{\textup{even}},\cdots,\ell_r E_{\textup{odd}})_{1,2} \ = \ C_{r,d}(\ell_1,\ldots,\ell_r).
\]

We now treat the case where $f$ has even degree, using an argument similar to the one form \ref{pro:usethetrace}.
By the above construction, for each $i$ we find
\[
\alpha(\tfrac{1}{d}\partial f/\partial x_i) \cdot E_{\textup{odd}} \simeq (\idtwo+\ell_{i,1} E_{\textup{odd}})
(\idtwo+\ell_{i,2} E_{\textup{even}})
\cdots
(\idtwo+\ell_{i,r_i} E_{\textup{odd}})
-\idtwo.
\]
We replace all $\eps$ by $\eps^3$, replace all $\alpha$ by $\eps$,
and lastly add $\idtwo$:
\[
\eps (\tfrac{1}{d}\partial f/\partial x_i) \cdot E_{\textup{odd}} 
+\idtwo
\equiv \big((\idtwo+\ell'_{i,1} E_{\textup{odd}})
(\idtwo+\ell'_{i,2} E_{\textup{even}})
\cdots
(\idtwo+\ell'_{i,r_i} E_{\textup{odd}})
\big) \pmod{\eps^3}.
\]
Analogously, when replacing $\alpha$ by $-\eps$ instead:
\[
-\eps (\tfrac{1}{d}\partial f/\partial x_i) \cdot E_{\textup{odd}} 
+\idtwo
\equiv \big((\idtwo+\ell''_{i,1} E_{\textup{odd}})
(\idtwo+\ell''_{i,2} E_{\textup{even}})
\cdots
(\idtwo+\ell''_{i,r_i} E_{\textup{odd}})
\big) \pmod{\eps^3}.
\]
We also find corresponding linear forms for the transposes.
Now observe that for any polynomials $a,b$ we have
\[
(-\eps a \cdot E_{\textup{odd}} + \idtwo + O(\eps^3))(-\eps b \cdot E_{\textup{even}} + \idtwo + O(\eps^3))
(\eps a \cdot E_{\textup{odd}} + \idtwo + O(\eps^3))(\eps b \cdot E_{\textup{even}} + \idtwo + O(\eps^3))
\]
\[
\equiv \begin{pmatrix}
1+ \eps^2 a\cdot b & 0 \\
0 & 1 - \eps^2 a\cdot b
\end{pmatrix} \pmod{\eps^3}.
\]
\begin{center}
\begin{tikzpicture}[scale=1.1]
\node at (-3,1.5) {Pictorially:};
\node[circle,fill=black,inner sep=0pt, minimum size = 0.5cm,minimum size=5pt] (01) at (0,1) {};
\node[circle,fill=black,inner sep=0pt, minimum size = 0.5cm,minimum size=5pt] (02) at (0,2) {};
\node[circle,fill=black,inner sep=0pt, minimum size = 0.5cm,minimum size=5pt] (11) at (1,1) {};
\node[circle,fill=black,inner sep=0pt, minimum size = 0.5cm,minimum size=5pt] (12) at (1,2) {};
\node[circle,fill=black,inner sep=0pt, minimum size = 0.5cm,minimum size=5pt] (21) at (2,1) {};
\node[circle,fill=black,inner sep=0pt, minimum size = 0.5cm,minimum size=5pt] (22) at (2,2) {};
\node[circle,fill=black,inner sep=0pt, minimum size = 0.5cm,minimum size=5pt] (31) at (3,1) {};
\node[circle,fill=black,inner sep=0pt, minimum size = 0.5cm,minimum size=5pt] (32) at (3,2) {};
\node[circle,fill=black,inner sep=0pt, minimum size = 0.5cm,minimum size=5pt] (41) at (4,1) {};
\node[circle,fill=black,inner sep=0pt, minimum size = 0.5cm,minimum size=5pt] (42) at (4,2) {};
\node[circle,fill=black,inner sep=0pt, minimum size = 0.5cm,minimum size=5pt] (51) at (5,1) {};
\node[circle,fill=black,inner sep=0pt, minimum size = 0.5cm,minimum size=5pt] (52) at (5,2) {};
\node[circle,fill=black,inner sep=0pt, minimum size = 0.5cm,minimum size=5pt] (61) at (6,1) {};
\node[circle,fill=black,inner sep=0pt, minimum size = 0.5cm,minimum size=5pt] (62) at (6,2) {};

\node (c) at (4.5,1.5) {$=$};
\draw (01) -- (41);
\draw (02) -- (42);
\draw (51) -- (61);
\draw (52) -- (62) node[inner sep=1.5pt,midway,above] {\scalebox{0.75}{$1+\eps^2 ab$}};
\draw (51) -- (61) node[inner sep=1.5pt,midway,above] {\scalebox{0.75}{$1-\eps^2 ab$}};

\draw (02) -- (11) node[inner sep=1.5pt,midway,fill=white] {$-\eps a$};
\draw (11) -- (22) node[inner sep=1.5pt,midway,fill=white] {$-\eps b$};
\draw (22) -- (31) node[inner sep=1.5pt,midway,fill=white] {$\eps a$};
\draw (31) -- (42) node[inner sep=1.5pt,midway,fill=white] {$\eps b$};
\node at (0.5,0.6) {\scalebox{0.75}{$+O(\eps^3)$}};
\node at (1.5,0.6) {\scalebox{0.75}{$+O(\eps^3)$}};
\node at (2.5,0.6) {\scalebox{0.75}{$+O(\eps^3)$}};
\node at (3.5,0.6) {\scalebox{0.75}{$+O(\eps^3)$}};
\draw (51) -- (61) node[inner sep=0.3cm,midway,below] {\scalebox{0.75}{$+O(\eps^3)$}};
\end{tikzpicture}
\end{center}

Let $M(c) := \begin{pmatrix}
1+ \eps^2 c & 0 \\
0 & 1 - \eps^2 c
\end{pmatrix}$.
Now note that
\[
(M(a_1 b_1)+O(\eps^3)) \cdots (M(a_n b_n)+O(\eps^3)) \equiv M(a_1 b_1 + a_2 b_2 + \cdots a_n b_n) \pmod{\eps^3}.
\]
Setting $a_i = x_i$ and $b_i = \frac 1 {d}\partial f/\partial x_i$,
and using Euler's homogeneous function theorem,
we obtain polynomially many linear forms $\ell_1,\ldots,\ell_r$ so that
\[
M(f) \equiv \big((\idtwo+\ell_1 E_{\textup{odd}})
(\idtwo+\ell_2 E_{\textup{even}})
\cdots
(\idtwo+\ell_r E_{\textup{even}})
\big) \pmod{\eps^3}
\]
Subtracting $\idtwo$ on both sides and taking the degree $d$ homogeneous part of the $(1,1)$ entry:
\[
\eps^2 f \ \equiv \ \underbrace{\nc e_{r,d}(\ell_1 E_{\textup{odd}},\ell_2 E_{\textup{even}},\cdots,\ell_r E_{\textup{even}})_{1,1}}_{=C_{r,d}(\ell_1,\ldots,\ell_r)} \pmod{\eps^3}
\]

We replace all $\eps$ by $\eps^{d/2}$, to get~$\eps^d f \ \,\equiv\, C_{r,d}(\ell'_1,\ldots,\ell'_r) \,\pmod{\eps^{ 3d/2}}$.
Therefore, $f \ \simeq \ C_{r,d}(\eps^{-1}\cdot \ell'_1,\ldots,\eps^{-1}\cdot \ell'_r)$. Both cases together prove that $C_{n,d}$ is $\VTF$-hard under homogeneous linear border projections.
The $\VQP$-hardness under quasipolynomial homogeneous linear border projections now follows from \ref{thm:VQTFVQTP}.
\end{proof}

\subsection{Converting Formulas to Circuits Over the Arity 3 Basis}
\label{subsec:convaritythree}
In this section we prove the following theorem.

\begin{theorem}
\label{thm:VPVTP}
$\VF \subseteq \VTP$.
\end{theorem}
\begin{proof}
Let $h \in \VF$, i.e., by Brent's depth reduction, $h$ has formulas of polynomial size and logarithmic depth.
We treat the homogeneous components $f$ of $g_n$ independently.
If $f$ is of even degree, observe that if $f$ has a formula of depth $\delta$, then $\partial f/\partial x_i$ has a formula of depth $2\delta$ (by induction, using the sum and product rules of derivatives), which by \ref{lem:inputhomogenize} implies the existence of an IHL formula of depth $O(\delta)$ (note that $\partial f/\partial x_i$ is homogeneous of odd degree). Now we apply the odd-degree argument below for each partial derivative independently.

Let $f$ be of odd degree.
As a first step we convert the IHL formula into an IHL formula for which at each gate either all even homogeneous components vanish or all odd homogeneous components vanish. The construction is similar to the \ref{lem:inputhomogenize} and works as follows.
We replace each gate $v$ by two gates $v_\textup{odd}$ and $v_\textup{even}$, where at $v_\textup{even}$ the sum of the even degree components is computed,
and at $v_\textup{odd}$ the sum of the odd degree components is computed.
Let $f = f_\textup{even}+_\textup{odd}$ be the decomposition of $f$ into the even homogeneous parts and the odd homogeneous parts.
$\big((f+g)_\textup{even},(f+g)_\textup{odd}\big) = (f_\textup{even}+g_\textup{even},f_\textup{odd}+g_\textup{odd})$ so a sum gate is replaced by two sum gates.
Moreover,
$\big((f\cdot g)_\textup{even},(f \cdot g)_\textup{odd}\big) = (f_\textup{even}\cdot g_\textup{even} + f_\textup{odd}\cdot g_\textup{odd},f_\textup{even}\cdot g_\textup{odd}+f_\textup{odd}\cdot g_\textup{even})$,
so a product gate is replaced by 4 product gates and 2 summation gates. Here we use that the depth was logarithmic.

We now convert such a formula to an IHL circuit with the same number of gates, but over the arity 3 basis.
This part is a bit subtle, and therefore we do it more formally below.
We replace each even degree gate $v$ that computes $g$ with a gate that computes $z \cdot g$, where $z$ is a dummy variable.
Addition gates are not changed. For product gates there are three cases.
\begin{itemize}
 \item A product gate $v$ of two odd-degree polynomials $f$ and $g$. By induction we have an IHL circuit over the arity 3 basis for $f$ and for $g$.
 We construct the arity 3 product $z \times f \times g$.
  \item A product gate $v$ that has an odd-degree polynomial $f$ at its child $w$, and that has an even-degree polynomial $g$ at its child $u$. By induction we have IHL circuits $C$ and $D$ over the arity 3 basis for $f$ and for $z g$, respectively.
  We take $C$ and $D$, delete all instances of $z$ in $D$, and feed there the output of $C$ instead. The resulting circuit computes $fg$.
  \item A product of an even-degree polynomial $f$ and an even-degree polynomial $g$. By induction we have IHL circuits $C$ and $D$ over the arity 3 basis for $zf$ and for $zg$, respectively. We take $C$ and $D$, delete all instances of $z$ in $D$, and feed there the output of $C$ instead. The resulting circuit computes $zfg$.
\end{itemize}
The size of the resulting circuit is less or equal to the size of the formula (even though the depth can increase in this construction).
\end{proof}

\begin{remark}
\label{rem:VFVTF}
Even when starting with a formula of logarithmic depth, the resulting circuit does not necessarily have logarithmic depth, hence we do not obtain $\VF=\VTF$. This is because in the second bullet point we rearrange the circuit structure when we replace $z$.
\end{remark}

\begin{remark}
We also do not get $\VP=\VTP$, because note that the replacements of $z$ in the second and third bullet point can only be done, because in a formula the outdegree of each gate is at most~1, i.e., we do not reuse computation results. After we replace $z$ by $f$ in a subcircuit that computes $zg$, the original subcircuit computing $zg$ will be gone and cannot be reused.
\end{remark}

\subsection{Valiant-Skyum-Berkowitz-Rackoff Over the Arity 3 Basis}
\label{subsec:VSBR}

\begin{theorem}\label{thm:VQTFVQTP}
$\VQTF = \VQTP$.
\end{theorem}
\begin{proof}
The entire argument is over the arity 3 basis and each homogeneous component is treated separately.
Given a size $s$ circuit that computes an odd-degree polynomial, we use \ref{thm:VSBR} below to obtain a circuit of size $\poly(s)$ and depth $O(\log^2(s))$ that computes the same polynomial.
We unfold the circuit to a formula of the same depth.
The size is hence $3^{O(\log^2(s))}=s^{O(\log s)}$.
If $s=n^{\polylog(n)}$, then $s^{O(\log s)} = n^{\polylog(n)}$ \footnote{$
(n^{\log^i(n)})^{\log^j(n^{\log^i(n)})} = n^{\log^{i+ij+j}(n)}
$
}.
The even-degree case is done by treating each partial derivative independently.
\end{proof}

Since we know that $\VQF = \VQBP = \VQP$ and $\VQTF = \VQF = \VQTP$,
the situation of \eqref{eq:inclusions} simplifies:
\begin{equation}\label{eq:quasipolyallfive}
\VQTF = \VQF = \VQBP = \VQP = \VQTP.
\end{equation}

The following \cref{thm:VSBR} is needed in the proof of \ref{thm:VQTFVQTP}. It lifts the classical Valiant-Skyum-Berkowitz-Rackoff \cite{vsbr83} circuit depth reduction to the arity 3 basis. The argument is an adaption of the original argument.

\begin{theorem}[VSBR depth reduction for IHL circuits over the arity 3 basis]\label{thm:VSBR}
Let $f$ be a polynomial computed by a graded IHL circuit of size $s$ over the arity 3 basis, $\deg(f)=d$. Then there exists a graded IHL circuit over the arity 3 basis of size $O(\poly(s))$ and depth $O(\log(s)\cdot\log d)$ computing $f$.
\end{theorem}
\begin{proof}
We adapt the proof from \cite{Sap19}.
We treat only the homogeneous odd case, because all summands can be treated independently, and in the even degree case we can treat each partial derivative independently.
We work entirely over the arity 3 basis (and hence compute a polynomial whose even degree homogeneous parts all vanish), so every circuit and subcircuit is over the arity 3 basis, and every product is of arity 3.

A circuit whose root is an arity 3 product gate is denoted by $x \times y \times z$. A circuit whose root is an arity 2 addition gate is denoted by $x+y$, just as usual.
Notationally, we use the same notation for gates, for their subcircuits, and for the polynomials they compute.
If we want to specifically highlight that we talk about the circuit with root $w$, then we write $\langle w \rangle$.
We write $v \leq u$ is $v$ is contained in the subcircuit with root $u$.
We write $C \equiv C'$ to denote that the circuits $C$ and $C'$ compute the same polynomial.

Let $z$ be a new dummy variable.
Let the circuit $[u:v]$ be defined via
$[u:v] := z$ if $u=v$, and if $u\neq v$ we have
\[
[u:v]
\ := \ 
\begin{cases}
0 & \text{ if $u$ is a leaf}
\\
[u_1:v] + [u_2:v]& \text{ if } u = u_1+ u_2
\\
[u_1:v] \times u_2 \times u_3 &
\begin{minipage}[t]{7.5cm}if $u = u_1 \times u_2 \times u_3$ and $u_1$ has the highest degree among $\{[u_1],[u_2],[u_3]\}$
\end{minipage}
\end{cases}
\]
It can be seen by induction that $[u:v]$ is zero or a homogeneous polynomial of degree $\deg u - \deg v + 1$,
and $[u:v]$ is zero or is homogeneous linear in~$z$.
If $w \not\leq u$, then $[u:w]=0$.
For a circuit $C$ we write $[u:v]_{C} := [u:v](z \leftarrow C)$, where $\leftarrow$ means that all leaves labelled $z$ are replaced by the output of the circuit $C$.

We define a set of gates that is called the \emph{$m$-frontier} $\mathcal F_m$ via\\
$
\mathcal F_m := \{u \mid u = u_1 \times u_2 \times u_3 \ \textup{ with} \ \deg u_1, \deg u_2, \deg u_3 \leq m \ \textup{ and} \ \deg(u) > m\}
.
$
\begin{lemma}\label{lem:usum}
Fix a pair $(u,m)$ with $\deg u> m$.
Let $\mathcal F := \mathcal F_m$. Then
$
u \equiv \sum_{w \in \mathcal F} [u:w]_{\langle w\rangle}.
$
\end{lemma}
\begin{proof}
For the proof we fix $m$ and do induction on the \emph{depth} of $u$, i.e., the position of $u$ in any fixed topological ordering of the gates.
Since for every gate $u$ with $\deg(u)>m$ there exists some gate $u'\in\mathcal F \cap \langle u \rangle$, the induction start is the case $u \in \mathcal F$.
In this case, since $\mathcal F$ is an antichain, it follows that 
$\sum_{w \in \mathcal F}[u:w] = 0+[u:u] = z$, and hence $\sum_{w \in \mathcal F}[u:w]_{\langle w\rangle} = [u:u]_{\langle u \rangle} = z_{\langle u \rangle} = u$.
This proves that case $u \in \mathcal F$.
Now, let $u \notin \mathcal F$.
If $u$ is an addition gate:
\begin{align*}
u
&=&
u_1 + u_2
 \ \stackrel{\textup{I.H.}}{\equiv} \ 
\sum_{w \in \mathcal F} [u_1:w]_{\langle w\rangle} + \sum_{w \in \mathcal F} [u_2:w]_{\langle w\rangle}
h\ \equiv \ 
\sum_{w \in \mathcal F}\bigg( [u_1:w]_{\langle w\rangle} + [u_2:w]_{\langle w\rangle}\bigg)
\\
&=&
\sum_{w \in \mathcal F}\bigg( [u_1:w] + [u_2:w]\bigg)_{\langle w\rangle}
\ \stackrel{\textup{Def.}}{=} \ 
\sum_{w \in \mathcal F}[u:w]_{\langle w\rangle}
\end{align*}

If $u$ is a multiplication gate, note that $u \notin \mathcal F$, so one of the children has degree $>m$ (w.l.o.g.\ that child is called $u_1$):
\begin{align*}
u
&=&
u_1 \times u_2 \times u_3
\ \stackrel{\textup{I.H.}}{\equiv} \ 
\left(\sum_{w \in \mathcal F} [u_1:w]_{\langle w\rangle}\right)
\times
u_2 \times u_3
\  \equiv \ 
\sum_{w \in \mathcal F}\bigg( [u_1:w]_{\langle w\rangle}\times u_2 \times u_3\bigg)
\\
&=&
\sum_{w \in \mathcal F}\bigg( [u_1:w]\times u_2\times u_3\bigg)_{\langle w\rangle}
\ \stackrel{\textup{Def.}}{=} \ 
\sum_{w \in \mathcal F}[u:w]_{\langle w\rangle}
\end{align*}

\vspace{-0.7cm}
\end{proof}

\begin{lemma}\label{lem:uvsum}
Fix a pair $(u,m,v)$ with $\deg u > m \geq \deg v$.
Let $\mathcal F := \mathcal F_m$.
\[
[u:v] \equiv \sum_{w \in \mathcal F} [u:w]_{[w:v]}.
\]
\end{lemma}
\begin{proof}
For the proof we fix $m$ and $v$ and do induction on the \emph{depth} of $u$, i.e., the position of $u$ in any fixed topological ordering of the gates.
Since for every gate $u$ with $\deg(u)>m$ there exists some gate $u'\in\mathcal F \cap \langle u \rangle$, the induction start is the case $u \in \mathcal F$.
In this case, since $\mathcal F$ is an antichain, it follows that $\sum_{w \in \mathcal F}[u:w]_{[w:v]} \equiv z_{[u:v]} = [u:v]$.
This proves that case $u \in \mathcal F$.
Now, let $u \notin \mathcal F$.
Since $\deg u>m$ and $m\geq \deg v$ we have $u \neq v$.
If $u$ is an addition gate:

\begin{align*}
[u:v]
&\stackrel{\textup{Def.\ $(u\neq v)$}}{=}&
[u_1:v] + [u_2:v]
\ \stackrel{\textup{I.H.}}{\equiv} \ 
\sum_{w \in \mathcal F} [u_1:w]_{[w:v]} + \sum_{w \in \mathcal F} [u_2:w]_{[w:v]}
\\
&\equiv&
\sum_{w \in \mathcal F}\bigg( [u_1:w]_{[w:v]} + [u_2:w]_{[w:v]}\bigg)
\ = \ 
\sum_{w \in \mathcal F}\bigg( [u_1:w] + [u_2:w]\bigg)_{[w:v]}
\\
&\stackrel{\clap{\scriptsize\textup{Def.}}}{=}&
\sum_{w \in \mathcal F}[u:w]_{[w:v]}
\end{align*}
If $u$ is a multiplication gate, note that $u \notin \mathcal F$, so one of the children has degree $>m$ (w.l.o.g.\ that child is called $u_1$):

\begin{align*}
[u:v]
&\stackrel{\textup{Def.\ $(u\neq v)$}}{=}&
[u_1:v] \times u_2 \times u_3
 \ \stackrel{\textup{I.H.}}{\equiv} \ 
\left(\sum_{w \in \mathcal F} [u_1:w]_{[w:v]}\right)
\times
u_2\times u_3
\\
&\equiv&
\sum_{w \in \mathcal F}\bigg( [u_1:w]_{[w:v]}\times u_2\times u_3\bigg)
\ = \ 
\sum_{w \in \mathcal F}\bigg( [u_1:w]\times u_2\times u_3\bigg)_{[w:v]}
\\
&\stackrel{\clap{\scriptsize\textup{Def.}}}{=}&
\sum_{w \in \mathcal F}[u:w]_{[w:v]}
\end{align*}

\vspace{-0.7cm}
\end{proof}
We now construct the shallow circuit so that the degree of each child in a multiplication gate decreases from $\delta$ to $\lceil \frac 2 3 \delta \rceil$, so the multiplication depth (i.e., the number of multiplications on a path from leaf to root) is at most $O(\log d)$.
Here we allow arity 5 multiplication gates. These can be simulated by two arity 3 multiplication gates.
We construct the circuit by induction on the degree, and we construct it in a way that each $u$ and each $[u:w]_{\langle v\rangle}$ are computed at some gate, so the size of the resulting circuit is at most $O(s^3)$.
The addition gates between the multiplications can be balanced, so that we have at most $O(\log s)$ depth in each addition tree.
This gives a total depth of $\log d \cdot \log s$.
\subsection{The construction for {\textit{u}}.}

\begin{align*}
u & \;\stackrel{\textup{Lem.}~\oldref{lem:usum}}{\equiv} \;\sum_{w\in\mathcal{F}}[u:w]_{\langle w\rangle}\; =\; \sum_{w\in\mathcal{F}}[u:w]_{\langle w_{1}\rangle}\times w_{2}\times w_{3}\ \\
 & \;=\; \sum_{\substack{w\in\mathcal{F} \\ \deg(u)\geq\deg(w)}}[u:w]_{\langle w_{1}\rangle}\times w_{2}\times w_{3}  \;\equiv\;\sum_{\substack{w\in\mathcal{F} \\ \deg(u)\geq\deg(w)}}[u:w]_{\langle w_{3}\rangle}\times w_{2}\times w_{1}
\end{align*}

This explicit rearrangement of $w_1$ and $w_3$ is necessary and goes beyond \cite{vsbr83}.
Choose
$m = \lceil\frac 2 3 \deg u\rceil$. Recall $\deg w_i \leq m$, so we already have two of the three cases: $\deg w_1 \leq \lceil\frac 2 3 \deg u\rceil$ and $w_2 \leq \lceil\frac 2 3 \deg u\rceil$.
But we also know $\deg(u)\geq \deg(w)=\deg(w_1)+\deg(w_2)+\deg(w_3)$, hence w.l.o.g.\ $\deg(w_3)\leq \lfloor\frac 1 3\deg(u)\rfloor$.
Therefore
$\deg u - \deg w + \deg w_3 \leq \lfloor\frac 4 3\rfloor \deg u - \underbrace{\deg w}_{>m} < \frac 2 3 \deg u$.

\subsection{The construction for [\textit{u:v}].}
We use fractions and ``$\cdot$'' multiplication signs when we do not have a circuit implementation in the intermediate equalities on polynomials. We write $w = w_1 \times w_2 \times w_3$ for $w \in \mathcal F$.
\begin{align*}
[u:v]
&\stackrel{\textup{Lem.}~\oldref{lem:uvsum}}{\equiv}
\sum_{w \in \mathcal F} [u:w]_{[w:v]}
\ = \ 
\sum_{\substack{w \in \mathcal F \\ \deg(u)\geq\deg(w)}} \frac{[u:w]}{z} \cdot [w:v] \\ &
 =  
\frac 1 z \sum_{\substack{w \in \mathcal F \\ \deg(u)\geq\deg(w)}} [u:w] \cdot [w_1:v] \cdot w_2 \cdot w_3
\equiv
\sum_{\substack{w \in \mathcal F \\ \deg(u)\geq\deg(w)}} [u:w]_{\langle w_3\rangle} \times [w_1:v] \times w_2
\\
 &\stackrel{\clap{\scriptsize\textup{Lem.}~\oldref{lem:usum}}}{\equiv}
\sum_{\substack{w \in \mathcal F \\ \deg(u)\geq\deg(w)}} [u:w]_{\langle w_3\rangle} \times [w_1:v] \times
\left(
\sum_{\substack{y \in \mathcal F' \\ \deg(w_2)\geq\deg(y)}}
[w_2:y]_{\langle y_3\rangle}
\times
y_2
\times
y_1
\right)
\\
&\equiv
\sum_{\substack{w \in \mathcal F \\ \deg(u)\geq\deg(w)}}
\sum_{\substack{y \in \mathcal F' \\ \deg(w_2)\geq\deg(y)}}
[u:w]_{\langle w_3\rangle} \times [w_1:v] \times
\big(
[w_2:y]_{\langle y_3\rangle}
\times
y_2
\times
y_1
\big)
\end{align*}
We set $m=\lceil\frac 2 3 (\deg u+\deg v)\rceil$ and $m'=\lceil\frac 2 3 \deg w_2\rceil$.
We calculate the degrees of the five factors:
\begin{itemize}
\item $\deg u\!-\!\deg w\!+\!\deg w_3 \leq (\deg u - \deg w) + \lfloor\frac 1 3 \deg u\rfloor \leq \lfloor\frac 4 3 \deg u\rfloor - m \leq \lceil\frac 2 3 (\deg u - \deg v)\rceil$
\item $\deg w_1 - \deg v +1 \leq \deg w_1 \leq m \leq \lceil\frac 2 3 (\deg u-\deg v)\rceil$
\item $\deg w_2 - \deg y + \deg y_3 \leq \lfloor\frac 4 3 \deg w_2\rfloor - \lceil\frac 2 3 \deg w_2\rceil \leq \lceil\frac 2 3 \deg w_2\rceil \leq \lceil\frac 2 3 (\deg u - \deg v)\rceil$
\item $\deg y_2 \leq \lceil\frac 2 3 \deg w_2\rceil \leq \lceil\frac 2 3 (\deg u - \deg v)\rceil$, and analogously for $\deg y_1$.
\end{itemize}
The rescaling constants on the edges can be set in the straightforward way.
\end{proof}

\end{document}